\documentclass[12pt]{article}
\pdfoutput=1
\usepackage[square,comma,numbers,sort&compress]{natbib}
\usepackage{graphicx,epstopdf,amssymb,amsfonts,amsmath,amsthm,array,
mathrsfs,amscd}
\usepackage[normalem]{ulem}
\usepackage{float}
\usepackage{xcolor}
\usepackage{appendix}
\usepackage{hyperref}
\newcommand{\pbs}[1]{\let\temp=\\#1\let\\=\temp}
\numberwithin{equation}{section}
\def\be{\begin{equation}}\def\ee{\end{equation}}
\def\cvp{\raise 2pt\hbox{,}}

\def\re{\mathop{\text{Re}}\nolimits}  \def\res{\mathop{\text{res}}\nolimits}
\def\arctanh{\mathop{\text{arctanh}}\nolimits}

 \def\d{{\rm d}} 
\def\la{\lambda}

\def\la{\lambda}

\def\disk{\mathscr D}

\newtheorem{lemma}{Lemma}[section]
\newtheorem{proposition}{Proposition}[section]
\newtheorem{theorem}{Theorem}[section]
\def\plb#1#2#3{{\it Phys.\ Lett.\ }{\bf B #1} (#2) #3}
\def\npb#1#2#3{{\it Nucl.\ Phys.\ }{\bf B #1} (#2) #3}

\def\jhep#1#2#3{{\it J. High Energy Phys.\ }{\bf #1} (#2) #3}
\def\prd#1#2#3{{\it Phys.\ Rev.\ }{\bf D #1} (#2) #3}

\def\cmp#1#2#3{{\it Comm.\ Math.\ Phys.\ }{\bf #1} (#2) #3}

\def\ap#1#2#3{{\it Ann.\ of Phys.\ }{\bf #1} (#2) #3}

\def\imath#1#2#3{{\it Invent math }{\bf #1} (#2) #3}

\def\am#1#2#3{{\it Advances in Math.\ }{\bf #1} (#2) #3}
\begin{document}
%
%
{\pagestyle{empty}
\parskip 0in
\

\vfill
\begin{center}
{\LARGE Dirichlet Scalar Determinants On}

\bigskip

{\LARGE Two-Dimensional Constant Curvature Disks}

\vspace{0.4in}


Soumyadeep C{\scshape haudhuri} and Frank F{\scshape errari}

\medskip
{\it Service de Physique Th\'eorique et Math\'ematique\\
Universit\'e Libre de Bruxelles (ULB) and International Solvay Institutes\\
Campus de la Plaine, CP 231, B-1050 Bruxelles, Belgique}

%

\smallskip
{\tt chaudhurisoumyadeep@gmail.com, frank.ferrari@ulb.be}
\end{center}
\vfill\noindent

We compute the scalar determinants $\det(\Delta+M^{2})$ on the two-dimensional round disks of constant curvature $R=0$, $\mp 2$, for any finite boundary length $\ell$ and mass $M$, with Dirichlet boundary conditions, using the $\zeta$-function prescription. When $M^{2}=\pm q(q+1)$, $q\in\mathbb N$, a simple expression involving only elementary functions and the Euler $\Gamma$ function is found. Applications to two-dimensional Liouville and Jackiw-Teitelboim quantum gravity are presented in a separate paper.

\vfill

\medskip
%
%
\newpage\pagestyle{plain}
\baselineskip 16pt
\setcounter{footnote}{0}

}

\tableofcontents

\section{Introduction}

Motivated by applications in holography at finite cut-off and two-dimensional quantum gravity models on finite geometries \cite{Loopcalc,JTferra}, we compute the partition function of a scalar field of mass $M$ on a round disk of finite boundary length $\ell$, with Dirichlet boundary conditions, in zero, negative and positive constant curvatures, that is to say, for a Ricci scalar $R=2\eta/L^{2}$, $\eta=0$, $-1$ or $+1$. This amounts to evaluating the determinants
\be\label{Ddef}\mathscr D^{\eta}(M^{2},\ell,L) = {\det}_{\text D}\bigl(\Delta_{\eta} + M^{2}\bigr)\ee
where $\Delta_{\eta}$ is the positive Laplacian for the appropriate disk metrics.  Explicitly, parameterizing the disk in polar coordinates $(r,\theta)$, $0\leq r\leq 1$, we have
\be\label{Deltadef} \Delta_{\eta}f =-e^{-2\sigma_{\eta}}\biggl(\frac{1}{r}\partial_{r}\bigl(r\partial_{r}f\bigr) + \frac{1}{r^{2}}\partial_{\theta}^{2}f\biggr)\ee
for a conformal factor
\be\label{conffactor} e^{\sigma_{0}} =\frac{\ell}{2\pi}\,\cvp\quad  e^{\sigma_{\pm}} = \frac{2r_{0}L}{1\pm r_{0}^{2}r^{2}}\,\cdotp\ee
The dimensionless parameter $r_{0}$ is related to the boundary length, boundary extrinsic curvature and bulk area according to 
\be\label{ellminrel} \ell^{\eta}(r_{0}) = \frac{4\pi r_{0}L}{1+\eta r_{0}^{2}}\,\cvp\quad 
k^{\eta}(r_{0}) = \frac{1-\eta r_{0}^{2}}{2r_{0}L}\,\cvp\quad
A^{\eta}(r_{0}) = \frac{4\pi  r_{0}^{2}L^{2}}{1+\eta r_{0}^{2}} = L\ell^{\eta} r_{0}\, .\ee
In negative curvature, $0\leq r_{0}<1$. In positive curvature, $r_{0}$ is an arbitrary positive real number. In this case, the two disks associated with $r_{0}$ and $1/r_{0}$ have the same boundary length but will of course yield two different determinants. The notation $\mathscr D^{+}(M^{2},\ell,L)$ is thus ambiguous and, when necessary, we use the more precise notation $\mathscr D^{+>}$ and $\mathscr D^{+<}$ to distinguish between $r_{0}\geq 1$ and $r_{0}\leq 1$. Eventually, it will be more convenient to use the parameter $r_{0}$ instead of $\ell$, see Eq.\ \eqref{deletadef}.

Let us briefly note that, instead of \eqref{Ddef}, one may prefer to study the dimensionless determinant ${\det}_{\text D}[\mu^{-2}(\Delta_{\eta} + M^{2})]$, where $\mu$ is an arbitrary renormalization scale. The equations \eqref{Ddef}, \eqref{Deltadef} and \eqref{ellminrel} show that $\mu$ can be trivially absorbed by rescaling $\ell$, $L$ and $M$ appropriately. We thus set $\mu=1$ in the following and work with dimensionless parameters.

Functional determinants like \eqref{Ddef} needs to be renormalized and are thus always defined modulo the addition of local counterterms. The possible counterterms here are of the form
\be\label{counterterm} \int\!\d^{2}x\sqrt{g}\, ,\quad \int\!\d^{2}x\sqrt{g}\, R\, ,\quad \oint\!\d s\, ,\quad \oint\!\d s\, k\, ,\ee
where the integrals are taken over the disk bulk or the disk boundary. In constant curvature, the two bulk terms are equivalent and yield the area (cosmological constant) counterterm. Due to the Gauss-Bonnet formula, the integral of the extrinsic curvature $k$ over the boundary also yields the area and a $\ell$-independent constant. Overall, $\ln \mathscr D_{\eta}$ is thus defined modulo the addition of terms of the form
\be\label{Dambiguity} c_{0}(M,L) + c_{1}(M,L)\ell + c_{2}(M,L)A\, .\ee
When one uses a specific regularization and renormalization scheme, these terms take a specific form, but they are not physical. 

We shall use the $\zeta$-function scheme \cite{zetaref,bfref} in the following. If $(\la_{p}^{\eta})_{p\geq 0}$ denotes the set of eigenvalues of the Laplacian, $\la^{\eta}_{0}\leq\la^{\eta}_{1}\leq\cdots$, the $\zeta$-function is defined by 
\be\label{zetadef}\zeta^{\eta}(s;M^{2},\ell,L) = \sum_{p\geq 0}\frac{1}{\bigl(\la_{p}^{\eta}+M^{2}\bigr)^{s}}\,\cdotp\ee
The series converges for $\re s>1$ and defines a meromorphic function on the complex $s$-plane, that is regular at $s=0$. The determinant is then defined as
\be\label{detzetadef}\mathscr D^{\eta}(M^{2},\ell,L) = e^{-(\zeta^{\eta})'(0;M^{2},\ell,L)}\, ,\ee
where the prime denotes the derivative with respect to $s$. This procedure may seem abstract, but one can show explicitly that it is equivalent to subtracting local counterterms from the effective action $\ln \det\bigl(\Delta + M^{2}\bigr)$, see e.g.\ \cite{bfref} for a pedagogical discussion.

Let us note that the flat case can be obtained as the limit $L\rightarrow\infty$ of either the positive (with the choice $r_{0}<1$) or the negative curvature cases,
\be\label{flatfromcurved} \mathscr D^{0}(M^{2},\ell) = \lim_{L\rightarrow +\infty}\mathscr D^{-}(M^{2},\ell,L)=\lim_{L\rightarrow +\infty}\mathscr D^{+<}(M^{2},\ell,L)\, .\ee
Nevertheless, it will be useful to treat the flat case separately.

The definition \eqref{detzetadef} allows to immediately eliminate one parameter from our problem. In zero curvature, we use
\be\label{dimzero}\zeta^{0}\bigl(s;M^{2},\ell\bigr) = \Bigl(\frac{\ell}{2\pi}\Bigr)^{2s}\zeta^{0}\biggl(s;\Bigl(\frac{\ell M}{2\pi}\Bigr)^{2},2\pi\biggr)\ee
%
to get
\begin{multline}\label{dimzerodet} \mathscr D^{0}\bigl(M^{2},\ell\bigr) = \Bigl(\frac{\ell}{2\pi}\Bigr)^{-2\zeta^{0}(0;(\ell M/(2\pi))^{2},2\pi)} \mathscr D^{0}\biggl(\Bigl(\frac{\ell M}{2\pi}\Bigr)^{2},2\pi\biggr)\\=
\Bigl(\frac{\ell}{2\pi}\Bigr)^{-\frac{1}{3}+\frac{AM^{2}}{2\pi}}
\mathscr D^{0}\biggl(\Bigl(\frac{\ell M}{2\pi}\Bigr)^{2},2\pi\biggr)\, .
\end{multline}
We have used the general formula
\be\label{zetaatzero} \zeta(0)=\frac{\chi}{6} -\frac{A M^{2}}{4\pi}=\frac{1}{6}-\frac{A M^{2}}{4\pi}\ee
for the value of the $\zeta$ function at $s=0$ in terms of the Euler characteristics $\chi$, which is one for the disk, and the area $A$ of the surface on which the Laplacian is defined. In the flat case, we thus work from now on with $\ell=2\pi$ and note
\be\label{delzerodef} D^{0}(M^{2})=\mathscr D^{0}(M^{2},2\pi)\, .\ee
In non-zero curvature, we use
\be\label{dimcurv}\zeta^{\eta}\bigl(s;M^{2},\ell,L\bigr) = L^{2s}\zeta^{\eta}\bigl(s;(LM)^{2},\ell/L,1\bigr)\ee
to get
\be\label{dimcurvdet} \mathscr D^{\eta}\bigl(M^{2},\ell,L\bigr) = L^{-2\zeta^{\eta}(0; (L M)^{2},\ell/L,1)} \mathscr D^{\eta}\bigl((L M)^{2},\ell/L,1\bigr)=L^{-\frac{1}{3} + \frac{A^{\eta}M^{2}}{2\pi}} \mathscr D^{\eta}\bigl((LM)^{2},\ell/L,1\bigr)\, .\ee
From now on, we thus work with $L=1$ and note
\be\label{deletadef} D^{\eta}(M^{2},r_{0})=\mathscr D^{\eta}(M^{2},\ell,1)\, .\ee
We shall use the parameter $r_{0}$ instead of $\ell$ to waive the ambiguity associated with the fact that $\ell(r_{0}) = \ell(1/r_{0})$ in the positive curvature case.

Our main goal in this work will be to derive the following.
\begin{theorem}
The functions $D^{0}(M^{2})$ and $D^{\eta}(M^{2},r_{0})$ are given by
\begin{align}\label{detzero} 
\begin{split}
&\ln D^{0}(M^{2}) = \frac{1}{3}\ln 2 - \frac{1}{2}\ln(2\pi)-\frac{5}{12}-2\zeta_{\text R}'(-1) +\frac{1}{2}\bigl(\gamma -1-\ln 2 \bigr) M^{2}\\&\hskip 8cm
+\sum_{n\in\mathbb Z}\ln\biggl[\frac{2^{|n|}|n|!}{M^{|n|}}I_{|n|}(M)
e^{-\frac{M^{2}}{4(|n|+1)}}\biggr]\, ,
\end{split}
\\
 \label{deteta}
\begin{split}
&\ln D^{\eta}(M^{2},r_{0}) = - \frac{1}{2}\ln(2\pi)-\frac{5}{12}-2\zeta_{\text R}'(-1) + \frac{\eta A^{\eta}(r_{0})}{3\pi}-\frac{1}{3}\ln r_{0}\\
& +\frac{1}{2\pi}\Bigl(\gamma -1 +\ln r_{0} \Bigr) A^{\eta}(r_{0}) M^{2}
+\sum_{n\in\mathbb Z}\ln\biggl[f_{n}^{\eta}(1;-M^{2})
e^{-\frac{A^{\eta}M^{2}}{4\pi(|n|+1)}F(1,1,|n|+2,\frac{\eta A^{\eta}}{4\pi})}\biggr]\, ,
\end{split}
\end{align}
where $\zeta_{\text R}$ is the Riemann $\zeta$-function, $\gamma$ is Euler's constant, $I_{n}$ is the modified Bessel function of the first kind, the functions $f_{n}^{\eta}$ are hypergeometric functions defined in Eq.\ \eqref{eigenminus} and \eqref{eigenplus}, and $A^{\eta}$ is the area, defined in \eqref{ellminrel}.
\end{theorem}
The above convergent series representation can be used to evaluate the determinants numerically with very good accuracy. For instance, in flat space, one can use the integral representation of the modified Bessel function,
\be\label{Besselintrepnum} I_{n}(z) = \frac{(z/2)^{n}}{\sqrt{\pi}\Gamma(n+1/2)}\int_{0}^{\pi}\!\d t\, e^{z\cos t}(\sin t)^{2n}\, ,\ee
to evaluate numerically $n!I_{n}$.\footnote{This is much more efficient than evaluating directly $n!I_{n}$ in Mathematica, probably because it eliminates the huge factorials by dividing explicitly by $\Gamma (n+1/2)$.}

We do not know how compute the infinite sums in terms of known functions for arbitrary values of $M^{2}$ but, interestingly, this can be done for special values of $M^{2}$, which turn out to be important in specific quantum gravity applications \cite{Loopcalc}. The simplifications occur when
\be\label{specialmasses} M^{2} = M^{2}_{\eta, q}=-\eta q(q+1)\, , \quad q\in\mathbb N\, . \ee
Note that in negative curvature, these special values of the mass are the values for which the conformal dimension of the boundary operator associated with the massive scalar field in hyperbolic space are integers,
\be\label{confdim} \Delta = \frac{1}{2}\Bigl(1+\sqrt{1+4 M^{2}}\Bigr) = q+1\, .\ee
For these special values, the functions $f_{n}^{\pm}$ are polynomials of degree $q$ of their argument $\eta A^{\eta}/(4\pi)$. For $q=0$, the infinite sums are actually zero, term by term. For $q=1$, which is relevant in quantum gravity, we find
\begin{multline}\label{detspecial}\ln D^{\eta}\bigl(M^{2}=-2\eta,r_{0}\bigr) =  - \frac{1}{2}\ln(2\pi)-\frac{5}{12}-2\zeta_{\text R}'(-1) +\frac{\eta}{\pi}\Bigl(\frac{4}{3}-\ln r_{0}\Bigr)A^{\eta} \\ - \frac{1}{3}\ln r_{0} + \ln \bigl(1-\eta r_{0}^{2}\bigr) - \frac{3-\eta r_{0}^{2}}{1+\eta r_{0}^{2}}\ln\bigl(1+\eta r_{0}^{2}\bigr)
- 2\ln\Gamma\Bigl(\frac{2}{1+\eta r_{0}^{2}}\Bigr)\, .
\end{multline}
The explicit formula for any integer $q$ is given in Section \ref{specialSec}, Eq.\ \eqref{specialdetform}.

When the mass parameter vanishes, the determinants are of course well-known \cite{Weis}, because the calculation is then reduced to a straightforward application of the conformal anomaly \cite{Polyakov}. There is also a vast literature on the computation of massive determinants on compact Riemann surfaces, see e.g.\ \cite{compactdet,Voros}, and for constant curvature surfaces with geodesic boundaries \cite{geoddet,KMW}, making the link with the Selberg $\zeta$-function and Barnes double gamma function. The determinants have also been studied in infinite volume, especially in negative curvature, motivated by applications in holography \cite{hyperdet} (see also \cite{hyperdet2,hyperdet3} and the references therein). This will be reviewed below. However, to the best of our knowledge, it is the first time that explicit formulas are derived for massive scalar field determinants, or determinants of the same kind, on finite constant curvature round disks. These calculations, and possible generalizations, are relevant for studies of finite-size quantum field theory and finite cut-off holography. Our main motivation came from specific applications to two-dimensional Liouville and Jackiw-Teitelboim quantum gravities \cite{Loopcalc}. In the very special case of the positive curvature disk with geodesic boundary, which corresponds to the hemisphere, we can make the link with the results in \cite{KMW}, finding a perfect match.

The plan of the paper is as follows. We begin by briefly reviewing in Section \ref{infSec} the case of infinite area in zero and negative curvatures, in order to put our finite size results in perspective. In particular, it is emphasized that the large area limit of the finite size determinants does not match the known strictly infinite size formulas in negative curvature, an important and subtle effect associated with hyperbolic space. In Section \ref{genfinareaSec}, we explain the general set-up in finite area. The rotational invariance of the problem allows to decompose the spectrum in terms of Fourier modes, with associated Sturm-Liouville operators and determinants. We also study the large mass expansion of the determinants by using the heat kernel. In flat space, the large mass expansion is equivalent to the large area expansion, but this is not so in curved space. In Section \ref{zeromassSec}, we compute the zero mass determinants by using the conformal anomaly. In Section \ref{SLSec}, we compute the Sturm-Liouville determinants for any value of the mass parameter. In Section \ref{ratioSec}, it is explained how to compute the ratio between the massive and the massless determinants. This is then used in Sections \ref{flatSec} and \ref{curveSec} to derive Eqs.\ \eqref{detzero} and \eqref{deteta}. In section \ref{curveSec}, we also use Eq. \eqref{deteta}  to derive an exact formula for the determinant (with an arbitrary mass) on a hemisphere. Section \ref{specialSec} is devoted to the study of the special case $M^{2}= \mp q(q+1)$. We are able to derive very explicit expressions for the determinants in these cases, in terms of elementary functions and the Euler $\Gamma$ function, Eq.\ \eqref{specialdetform}. Applying these formulas, we obtain the large area limit of the determinants in negative curvature.

\section{\label{infSec}The case of infinite area}

\subsection{Euclidean space}

On the flat infinite Euclidean plane, the $\zeta$ function per unit area is given by
\be\label{zetaflatinf} \zeta^{0}_{\infty} (s;M^{2}) =\frac{1}{2\pi} \int_{0}^{\infty}\frac{\nu\d\nu}{(\nu^{2}+M^{2})^{s}} = \frac{M^{2(1-s)}}{4\pi(s-1)}\, \cdotp\ee
The measure $\frac{1}{2\pi}\nu\d\nu$ is a Euclidean ``Plancherel'' measure straightforwardly obtained by diagonalizing the Laplacian with plane waves and going to polar coordinates in momentum space. We get the determinant
\be\label{detflatinf} \ln D^{0}_{\infty} = \frac{M^{2}}{4\pi}\bigl(1-\ln M^{2}\bigr)\, .\ee

In order to understand the negative curvature case below, it is interesting to rederive this result by using a basis of eigenfunctions diagonalizing the rotation generator around the origin. One thus expands the plane waves in Fourier components
\be\label{Fourierflat} e^{i\vec p\cdot\vec x} = e^{i \nu r\cos(\theta-\phi)} =\sum_{n\in\mathbb Z}i^{n}J_{|n|}(\nu r)e^{in(\theta-\phi)}\, ,\ee
where $\theta$ and $\phi$ are the polar angles of $\vec x$ and $\vec p$ and $J$ the Bessel functions. This yields an orthonormal basis of eigenfunctions 
\be\label{orthbasisflat} \mathscr J_{\nu,n}(r,\theta) = \sqrt{\frac{\nu}{2\pi}}\,e^{in\theta}J_{|n|}(\nu r)\, .\ee
The normalization is chosen in such a way that the usual normalization condition for the plane waves implies
\be\label{normflat} \int_{0}^{2\pi}\!\d\theta\,\int_{0}^{1}\!\d r\, r \mathscr J^{*}_{\nu_{1},n_{1}}(r,\theta)\mathscr J_{\nu_{2},n_{2}}(r,\theta) = \delta(\nu_{1}-\nu_{2})\delta_{n_{1},n_{2}}\, .\ee
The Fourier expansion of a function $f$ in the basis $(\mathscr J_{\nu,n})_{\nu\geq 0,\,n\in\mathbb Z}$ then reads
\be\label{FourierJn} f(r,\theta) = \sum_{n\in\mathbb Z}\int_{0}^{\infty}\!\d\nu\, c_{n}(\nu) \mathscr J_{\nu,n}(r,\theta)\ee
with coefficients
\be\label{Fouriercoef} c_{n}(\nu) = \int_{0}^{2\pi}\!\d\theta\,\int_{0}^{1}\!\d r\, r\mathscr J^{*}_{\nu,n}(r,\theta) f(r,\theta)\, .\ee
In this formalism, the position-space $\zeta$ function reads
\be\label{zetflatpos} \zeta^{0}_{\infty} (s;M^{2};\vec x_{1},\vec x_{2}) = \sum_{n\in\mathbb Z}\int_{0}^{\infty}\!\d\nu\, \frac{\mathscr J^{*}_{\nu,n}(r_{1},\theta_{1})\mathscr J_{\nu,n}(r_{2},\theta_{2})}{(\nu^{2}+M^{2})^{s}}\, \cdotp\ee
Because Euclidean space is homogeneous, the $\zeta$ function at coinciding points $\vec x_{1}=\vec x_{2}$ is actually space-independent. Integrating over space then yields an infinite area factor. Factorizing this factor yields the zeta function per unit area
\be\label{zetapla} \zeta^{0}_{\infty}(s;M^{2}) = \int_{0}^{\infty}\frac{\mathfrak m_{\text E}(\nu)\d\nu}{(\nu^{2}+M^{2})^{s}}\, \cvp\ee
with an integration measure, the ``Euclidean Plancherel measure,'' given by
\be\label{Planflat} \mathfrak m_{\text E}(\nu) = \sum_{n\in\mathbb Z}\bigl|\mathscr J_{\nu,n}(0)\bigr|^{2} = \frac{\nu}{2\pi}\,\cdotp\ee
Of course, we find again the formula \eqref{zetaflatinf}.

\subsection{Hyperbolic space}

There is a strong formal similarity between the spectral problems for the Laplacian on the non-compact hyperbolic space and on the Euclidean space reviewed above. As in the Euclidean case, the spectrum is continuous and a basis of eigenfunctions, analogous to the plane waves, exists. Working in units for which $L=1$ and thus the curvature is $R=-2$, the eigenvalues of the Laplacian are parameterized as $\frac{1}{4}+\nu^{2}$ with $\nu\geq 0$. The fact that the eigenvalues are always greater than $1/4$ is the famous Breitenlohner-Freedman bound. The usual plane waves are replaced by non-Euclidean plane waves which are given explicitly by an appropriate power of the Poisson kernel. They can be expanded in Fourier components as
\be\label{FnonEuclexp} \biggl(\frac{1-r^{2}}{1+r^{2}-2 r \cos (\theta-\phi)}\biggr)^{\frac{1}{2}+i\nu}=\sum_{n\in\mathbb Z}F_{\nu,|n|}(r)e^{in(\theta-\phi)}
\ee
where
\be\label{defFn} F_{\nu,n}(r) = \frac{\Gamma(1/2+i\nu+n)}{n!\Gamma (1/2+i\nu)} r^{n}F\Bigl(\frac{1}{2}+i\nu,\frac{1}{2}-i\nu,n+1,-\frac{r^{2}}{1-r^{2}}\Bigr)\ee
is given in terms of a usual hypergeometric function. This yields an orthonormal basis 
\be\label{orthnonEucl}\mathscr F_{\nu,n}(r,\theta) = \mathscr N_{n}(\nu)e^{i n\theta}F_{\nu,|n|}(r)\, ,\ee
with a normalization factor $\mathscr N_{n}(\nu)$
chosen in such a way that 
\be\label{normhyper}\int_{0}^{2\pi}\!\d\theta\,\int_{0}^{1}\!\d r\, \frac{4r}{(1-r^{2})^{2}}  \mathscr F^{*}_{\nu_{1},n_{1}}(r,\theta)\mathscr F_{\nu_{2},n_{2}}(r,\theta) = \delta(\nu_{1}-\nu_{2})\delta_{n_{1},n_{2}}\, .\ee
The non-Euclidean Fourier expansion reads
\be\label{FnonEucl} f(r,\theta) = \sum_{n\in\mathbb Z}\int_{0}^{\infty}\!\d\nu\, c_{n}(\nu)\mathscr F_{\nu,n}(r,\theta)\ee
with non-Euclidean Fourier coefficients
\be\label{nonEFcoeff} c_{n}(\nu) = \int_{0}^{2\pi}\!\d\theta\,\int_{0}^{1}\!\d r\, \frac{4r}{(1-r^{2})^{2}}\mathscr F^{*}_{\nu,n}(r,\theta)f(r,\theta)\, .\ee
The $\zeta$ function per unit area is then given by
\be\label{zetahyper} \zeta^{-}_{\infty}(s;M) = \int_{0}^{\infty}\frac{\d\nu\,\frak m(\nu)}{(\nu^{2}+1/4+M^{2})^{s}}\, \cvp\ee
with an integration measure, called the Plancherel measure, given by
\be\label{Planhyper} \frak m(\nu) = \sum_{n\in\mathbb Z}\bigl|\mathscr F_{\nu,n}(0)\bigr|^{2} = \bigl|\mathscr N_{0}(\nu)\bigr|^{2}\, .\ee
There remains the problem of determining the normalization factor, so that \eqref{normhyper} is satisfied. This looks a priori non-trivial, because the integral involving the product of hypergeometric functions does not seem to be standard. However, the $F_{\nu,n}$ are eigenfunctions for the Sturm-Liouville operators $L_{n}^{-}$ defined in Eq.\ \eqref{flatSL} below, in the case $r_{0}=1$. A standard trick in Sturm-Liouville theory allows to compute the normalization integrals of the eigenfunctions in terms of their asymptotics. These asymptotics are well-known in the case of the hypergeometric function \eqref{defFn}, and thus the calculation can be done straightforwardly. Since this is not so important for our purposes, we will refrain from giving more details here. The result is the famous Plancherel measure on the hyperbolic space
\be\label{Plancherel} \frak m(\nu) = \frac{\nu\tanh (\pi\nu)}{2\pi}\,\cdotp\ee
From this, one may obtain a rather explicit formula for the determinant. In order to perform the analytic continuation of the $\zeta$ function to $s=0$, it is useful to write
\be\label{plantrick} \tanh(\pi\nu) = 1 - \frac{2e^{-2\pi\nu}}{1+e^{-2\pi\nu}}\,\cdotp\ee
This allows to isolate the non-holomorphic piece, which is entirely given by the simple pole at $s=1$,
\be\label{zetinf4} \zeta^{-}_{\infty}(s;M^{2}) = \frac{(M^{2}+1/4)^{1-s}}{4\pi(s-1)} - \frac{1}{\pi}\int_{0}^{\infty}\!\d\nu\,\frac{\nu}{e^{2\pi\nu}+1}\frac{1}{(M^{2}+\nu^{2}+1/4)^{s}}\, \cdotp\ee
This yields the logarithm of the infinite volume determinant, per unit area,
\be\label{dethyperinf} \ln D^{-}_{\infty} = \frac{M^{2}+1/4}{4\pi}\Bigl(1-\ln(M^{2}+1/4)\Bigr)-\frac{1}{\pi}\int_{0}^{\infty}\!\d\nu\,\frac{\nu}{e^{2\pi\nu}+1}\ln\bigl(M^{2}+\nu^{2}+1/4\bigr)\, .\ee

\noindent\emph{Comments:}

i) The integral in \eqref{dethyperinf} might be given a more explicit expression, see e.g.\ \cite{caldarelli}. The general result is that in even dimensions, the formulas are complicated, but in odd dimensions simple closed-form formulas can be found, because the Plancherel measures are then much simpler.

ii) Note that the expression \eqref{dethyperinf} does not have any direct physical significance, since it corresponds to an area counterterm. The same can of course be said in the flat space case.

iii) The Plancherel measure has an obvious interpretation as a density of eigenvalues per unit area of the Laplacian. More precisely, any naive discrete finite area expression involving a sum over the eigenvalues is replaced by an integral with the Plancherel measure,
\be\label{dosPlanhyper} \frac{1}{A}\sum_{p\geq 0}f(\la_{p}) \longrightarrow \int_{0}^{\infty}\!\d\nu\,\frak m(\nu) f(\nu^{2}+1/4)\, .\ee

iv) However, the interpretation of the Plancherel measure as the infinite area density of eigenvalues must be taken with extreme caution due to the following subtlety associated with hyperbolic space \cite{Adachi}.\footnote{We would like to thank Steve Zelditch for emphasizing this crucial property.} 

If one considers the problem for the disk of finite area and finite circumference $\ell$, as we are going to do below, the spectrum of the Laplacian is discrete. At this level, the situation is very similar to the case of flat space. One can define a density of eigenvalues per unit area
\be\label{density} \frak r(\la) = \frac{1}{A}\sum_{p\geq 0}\delta(\la-\la_{p})\ee
such that
\be\label{densint} \frac{1}{A}\sum_{p\geq 0}f(\la_{p}) = \int\!\d\lambda\, \frak r(\la)f(\la) = 2\int_{0}^{\infty}\!\d\nu\, \nu\, \frak r(\nu^{2}+1/4) f(\nu^{2}+1/4)\, . \ee
Note that the finite area density $\frak r$ depends, of course, on the type of boundary conditions we use; we focus in the present paper on Dirichlet boundary conditions. 

In the case of flat space, it is easy to check that the large $\ell$ limit of the density is independent of the boundary conditions and yields the limit measure $\frac{1}{2\pi}\nu\d\nu$, Eq.\ \eqref{Planflat}. But things are much more complicated in hyperbolic space. The qualitative origin of the difficulty is the well-known geometric property of hyperbolic space that, at large $\ell$, the area scales as the circumference, $A\sim\ell$, and, moreover, the area of an annulus of upper circumference $\ell$ and fixed geodesic width $a$ divided by the total area goes to a constant $1-e^{-a}$ instead of zero. As a consequence, the large $\ell$ limit is keeping track of the boundary condition used in finite area and the limit of the density of eigenvalues \eqref{densint} does not coincide with the Plancherel measure \eqref{Plancherel} \cite{Adachi}. The actual form of the limit of the density $\frak r$ is not known, even for simple classes of boundary conditions like Dirichlet or Neumann! This is an outstanding open problem. Our exact finite area results below will allow us to illustrate this subtle effect in a non-trivial way, see the end of Section \ref{zeromassSec} and especially Section \ref{specialSec2}.

\section{\label{genfinareaSec}Generalities on the cases of finite area}

As explained above, we set $\ell=2\pi$ in flat space and $L=1$ in curved space.

\subsection{\label{SLoneSec}Spectra and Sturm-Liouville operators}

The spectral problem for the Laplacian is studied by decomposing the eigenfunctions in Fourier components. This reduces the problem to computing the spectra of purely radial Sturm-Liouville operators, 
\be\label{flatSL} L_{n}^{\eta} =  -\frac{1}{\omega_{\eta}(r)^{2}}\Biggl[\frac{1}{r}\frac{\d}{\d r}\Bigl(r\frac{\d}{\d r}\,\cdot\Bigr) - \frac{n^{2}}{r^{2}}\Biggr]\, ,\ee
with
\be\label{omegaeta}\omega_{0}(r)=1\, ,\quad \omega_{\pm}(r) = \frac{2r_{0}}{1\pm r_{0}^{2}r^{2}}\, \cdotp\ee
With Dirichlet boundary conditions at $r=1$, the operators $L_{n}^{\eta}$ are symmetric with respect to the scalar product
\be\label{scaprodSL} \langle f_{1}|f_{2}\rangle = \int_{0}^{1} \omega_{\eta}(r)^{2}r f_{1}(r)f_{2}(r)\, \d r\, .\ee
Note that $\smash{L_{n}^{\eta}=L_{-n}^{\eta}}$, it is thus enough to consider $n\geq 0$. The spectra of the operators $\smash{L_{n}^{\eta}}$ are discrete, and we denote their eigenvalues by $\la^{\eta}_{n,k}$, $n\geq 0$, ordered such that $\smash{\la^{\eta}_{n,k+1}>\la^{\eta}_{n,k}}$. The eigenvalues are strictly positive, because the $\smash{L_{n}^{\eta}}$ are positive operators and the constant mode is projected out by the Dirichlet boundary conditions. In negative curvature, we actually know that all the eigenvalues are strictly greater than $1/4$ by the Breitenlohner-Freedman bound, see below.

A complete basis of orthogonal eigenfunctions is straightforward to find. One first seeks the solution of $L_{n}^{\eta}\cdot f = \la f$, for $\la>0$, that is regular at $r=0$. One finds
\be\label{eigenflat} f_{n}^{0}(r;\la) =  J_{|n|}\bigl(\sqrt{\la}\, r\bigr)\, ,\ee
where $J_{|n|}$ is the Bessel function of the first kind, in the flat case, and
\begin{align}\label{eigenminus} f_{n}^{-}(r;\la) &=  r^{|n|} F\Bigl(\frac{1}{2} + \frac{i}{2}\sqrt{4\la - 1},\frac{1}{2}-\frac{i}{2}\sqrt{4\la-1},|n|+1,-\frac{r_{0}^{2}r^{2}}{1- r_{0}^{2}r^{2}}\Bigr) ,\\\label{eigenplus}
f_{n}^{+}(r;\la) &= r^{|n|} F\Bigl(\frac{1}{2} + \frac{1}{2}\sqrt{4\la+1},\frac{1}{2}-\frac{1}{2}\sqrt{4\la+1},|n|+1,\frac{r_{0}^{2}r^{2}}{1+ r_{0}^{2}r^{2}}\Bigr) ,
\end{align}
where $F=_{2}\!\!F_{1}$ is the standard hypergeometric function, in the negative and positive curvature cases.\footnote{These eigenfunctions are not normalized.} One then imposes the Dirichlet condition,
\be\label{Dirichleteigen} f_{n}^{\eta}(1;\la) = 0\, ,\ee
whose solutions yield the full spectrum.

One can be a little bit more precise on the structure of the spectrum, by using the fact that the eigenvalues $\la^{\eta}_{n,k}(r_{0})$ are strictly decreasing functions of $r_{0}$. This intuitive result is derived explicitly in App.\ \ref{AppA}. In positive curvature, the limit $r_{0}\rightarrow\infty$ must yield the spectrum $\{p(p+1), p\in\mathbb N\}$ of the Laplacian on the round sphere.\footnote{One may worry that the limit yields the round sphere minus a point, but this is immaterial at the level of $L^{2}$-spaces. In particular, it is straightforward to check that the ground state wave function of the Dirichlet problem, which is the ground state of the operator $L_{0}^{+}$, converges with respect to the norm associated with the scalar product \eqref{scaprodSL} to the s-wave constant mode ground state of the round sphere problem.} In particular, all the eigenvalues must be strictly greater than 2, except the lowest $\la_{0,0}^{+}$.

The $\zeta$ functions \ref{zetadef} we wish to study can be conveniently written as
\be\label{zetaandSL} \zeta^{\eta}(s;M^{2}) =\sum_{p\geq 0}\frac{1}{\bigl(\la_{p}^{\eta}+M^{2}\bigr)^{s}}= \zeta^{\eta}_{0}(s;M^{2}) + 2\sum_{n\geq 1}\zeta^{\eta}_{n}(s;M^{2})\, ,\ee
where 
\be\label{zetaSL} \zeta^{\eta}_{n}(s;M^{2}) = \sum_{k\geq 0}\frac{1}{\bigl(\la_{n,k}^{\eta}+M^{2}\bigr)^{s}}\ee
is the $\zeta$-function for the operator $L^{\eta}_{n}+M^{2}$, with associated Sturm-Liouville determinants
\be\label{SLdetdef} d_{n}^{\eta} = \det \bigl(L_{n}^{\eta}+M^{2}\bigr) = e^{-(\zeta^{\eta}_{n})'(0)}\, .\ee
Note that one and two-dimensional Weyl's law, see e.g.\ \cite{Weyllaw}, yield the asymptotics of the eigenvalues as
\be\label{Weyllaw} \la^{\eta}_{n,k}\underset{k\rightarrow\infty}{\sim} \Bigl(\frac{\pi k}{a^{\eta}}\Bigr)^{2}\, , \quad \la^{\eta}_{p}\underset{p\rightarrow\infty}{\sim}\frac{4\pi p}{A^{\eta}}\, \cvp\ee
where
\be\label{aetadef} a^{\eta} = \int_{0}^{1}\!\omega_{\eta}(r)\,\d r = 
\begin{cases} 1 & \text{if $\eta=0$},\\ 2\arctanh r_{0} = \ln\frac{1+r_{0}}{1-r_{0}} & \text{if $\eta=-1$},\\ 2\arctan r_{0} & \text{if $\eta=+1$}
\end{cases}\ee
and $A^{\eta}$ is the area, see Eq.\ \eqref{ellminrel}. This shows that the series representations for $\zeta_{n}^{\eta}(s)$ and $\zeta^{\eta}(s)$ converge for $\re s>1/2$ and $\re s>1$, respectively.

Our strategy will be to discuss first the Sturm-Liouville determinants, in Section \ref{SLSec}, before dealing with the full problem.

\noindent\emph{Remark:}

The $\zeta$ functions are unambiguously defined when $M^{2}$ is chosen in such a way that $\la_{n,k}^{\eta}+M^{2}>0$. If this is not the case, the result for the determinant is obtained by analytically continuing in $M^{2}$. This analytic continuation is easy to perform by using the following observation. If one considers the $\zeta$ function defined in \eqref{zetadef} and if one sets $\zeta^{\eta,p_{0}}(s) = \sum_{p\geq p_{0}}(\la_{p}^{\eta}+M^{2})^{-s}$, then the determinant is such that
\be\label{factdet} D_{\eta} = e^{-(\zeta^{\eta})'(0)} = \Biggl(\prod_{p=0}^{p_{0}-1}\bigl(\la_{p}^{\eta}+M^{2}\bigr)\Biggr) e^{-(\zeta^{\eta,p_{0}})'(0)}\, .\ee
This formula makes the analytic continuation obvious by choosing $p_{0}$ such that $\la_{p}^{\eta} + M^{2}>0$ for all $p\geq p_{0}$.

\subsection{\label{hkflatSec}The large $M$ and large $\ell$ expansions}

\subsubsection{Flat case}

In the flat case, Eq.\ \eqref{dimzero}, \eqref{dimzerodet} show that the large $\ell$ limit is  equivalent to the large $M$ limit. But it is well-known that the large $M$ limit is simple and given in terms of the heat kernel expansion.

Indeed, we can write\footnote{The dependence in $\ell$ is kept implicit in our notation for the $\zeta$ function.}
\be\label{zetaKexp} \zeta^{0} (s) = \frac{1}{\Gamma(s)}\int_{0}^{\infty}\!\d t\, t^{s-1}e^{-M^{2}t}K^{0}(t)\, ,\ee
where $K^{0}$ is the heat kernel for the flat Laplacian. At large $M$, the integral \eqref{zetaKexp} is dominated by the small $t$ region in which we can use the expansion
\be\label{Kexpan1} K^{0}(t) = \frac{1}{4\pi t}\sum_{k\in\mathbb N/2}a_{k}t^{k}\, .\ee
This immediately yields the large $M$ asymptotic expansion of the $\zeta$ function
\be\label{zetaexpan1} \zeta^{0}(s) \underset{M\rightarrow\infty}{=} \frac{1}{4\pi\Gamma(s)}\sum_{k\in\mathbb N/2} a_{k}\Gamma(s-1+k)M^{2-2k-2s}\ee
and thus of
\begin{multline}\label{detexpan1} \ln \mathscr D^{0}(M^{2},\ell) 
\underset{M\rightarrow\infty}{=} \frac{M^{2}a_{0}}{4\pi}\bigl(1-\ln M^{2}\bigr) +
\frac{Ma_{1/2}}{2\sqrt{\pi}} + 
\frac{a_{1}}{2\pi}\ln M \\- \frac{1}{4\pi}\sum_{\substack{k\in\mathbb N/2\\k\geq 3/2}}a_{k}\Gamma(k-1)M^{2-2k}\, .
\end{multline}
The heat kernel coefficients for Dirichlet boundary conditions are well-known, see e.g.\ \cite{Gilkey,heatkernel}. The first coefficients for the flat disk\footnote{The first three coefficients are universal: $a_{0}$ is Weyl's law, $a_{1/2}$ gives the correction to Weyl's law in the presence of a boundary and $a_{1}$ is related to the Euler characteristics.} are
\be\label{hkcoeff} a_{0}= A=\frac{\ell^{2}}{4\pi}\, \cvp\quad a_{1/2} = -\frac{\sqrt{\pi}}{2}\ell\, ,\quad a_{1}=\frac{2\pi}{3}\, \cvp\quad a_{3/2} = \frac{\pi^{5/2}}{16\ell}\, \cvp\ee
which yields
\be\label{detexpan2} \ln D^{0}(M^{2}) = \frac{1}{4}M^{2}\bigl(1-\ln M^{2}\bigr)-\frac{\pi M}{2}+\frac{1}{3}\ln M - \frac{\pi}{128 M} + O\bigl(M^{-2}\bigr)\, .
\ee
Using Eq.\ \eqref{dimzerodet}, we get
\be\label{detexpan3} \ln \mathscr D^{0}(M^{2},\ell) = \Bigl(\frac{\ell M}{4\pi}\Bigr)^{2}
\bigl(1-\ln M^{2}\bigr)-\frac{\ell M}{4}+\frac{1}{3}\ln M - \frac{\pi^{2}}{64\ell M} + O\bigl((\ell M)^{-2}\bigr)\, .\ee
Per unit area, and at leading order, this matches with the infinite area result \eqref{detflatinf}.

\subsubsection{Curved cases}

In positive or negative curvature, the only change compared to the discussion in the previous subsection concerns the evaluation of the heat kernel coefficient $a_{3/2}$. The general formula yields \cite{heatkernel}
\be\label{athhkcoeff} a_{3/2}=-\frac{\sqrt{\pi}}{64}\oint\!\d s\,\bigl(4R-k^{2}\bigr)
= \frac{\sqrt{\pi}}{64}\,\ell\biggl[-9\eta +\Bigl(\frac{2\pi}{\ell}\Bigr)^{2}\biggr]\, .\ee
Using \eqref{dimcurvdet}, we get
\begin{multline}\label{dethkhyper} \ln\mathscr D^{\eta}(M^{2},\ell,L) = \frac{A^{\eta}M^{2}}{4\pi}\bigl(1-\ln M^{2}\bigr)-\frac{\ell M}{4}+\frac{1}{3}\ln M\\ - \frac{\pi^{2}}{64 \ell M}\Bigl(1-\frac{9\eta\ell^{2}}{4\pi^{2}L^{2}}\Bigr) + O\bigl(1/(LM)^{2}\bigr)\, .
\end{multline}
Curvature effects have generated the new term $-9\eta\ell^{2}/(4\pi^{2}L^{2})$. In negative curvature, this term shows explicitly that the large $\ell$ expansion is distinct from the large $M$ expansion. As already mentioned at the end of Section \ref{infSec}, the large $\ell$ limit will be much more difficult to study; see the end of the next Section and Section \ref{specialSec2} for further discussion.

\section{\label{zeromassSec}Zero-mass determinants}

For $M=0$, the determinants can be straightforwardly computed using the conformal anomaly. 

The conformal anomaly relates the determinants of the Laplacian for metrics in the same conformal class. In general, if $g$ and $g_{0}$ are two metrics such that $g = e^{2\sigma} g_{0}$, then
\be\label{conformalano}
\ln \det\Delta_{g} = \ln \det\Delta_{g_{0}} - \frac{1}{12\pi}\tilde S_{\text L}(g_{0},g)\, .
\ee
The ``total'' action $\tilde S_{\text L}$ is the sum $\tilde S_{\text L} =  S_{\text L} + \Delta S_{\text L}$ of the usual Liouville action
\be\label{Liouville} S_{\text L}(g_{0},g) = \int\!d^{2}x\sqrt{g_{0}}\,\bigl( g_{0}^{\mu\nu}\partial_{\mu}\sigma\partial_{\nu}\sigma + R_{0}\sigma\bigr) + 2\oint\!\d s_{0}\, k_{0}\sigma\ee
and of a correction term
\be\label{Liouvcorr}\Delta S_{\text L}(g_{0},g) = 3\oint\!\d s_{0}\, n_{0}^{\mu}\partial_{\mu}\sigma = 3\oint\Bigl(\d s\, k - \d s_{0} k_{0}\Bigr)\, .\ee
We note $R$, $k$ and $R_{0}$, $k_{0}$ the Ricci scalar and extrinsic curvature of the boundary for the metrics $g$ and $g_{0}$ respectively. The correction term $\Delta S_{\text L}$ is a counterterm contribution and is thus irrelevant for physics. We include it because we want to provide exact formulas in the $\zeta$-function scheme and this counterterm is precisely generated in the $\zeta$-function computation of the determinant with Dirichlet boundary conditions.

In the flat case, applying the above general formulas for $g_{0}$ and $g$ being the flat disk metrics of circumferences $2\pi$ and $\ell$ respectively, we get 
\be\label{zeromass1}\ln \mathscr D^{0}(0,\ell) = \ln\mathscr D^{0}(0,2\pi)-\frac{1}{3}\ln\frac{\ell}{2\pi} = \ln D^{0}(0) -\frac{1}{3}\ln\frac{\ell}{2\pi}\, \cdotp\ee
The determinant $\mathscr D^{0}(0,2\pi) = D^{0}(0)$ was computed in \cite{Weis,Spreafico},
\be\label{zeromassfinal} \ln D^{0}(0)  = \frac{1}{3}\ln 2 - \frac{1}{2}\ln(2\pi) -\frac{5}{12} - 2\zeta_{\text R}'(-1)\, .\ee
This is a physically irrelevant and scheme-dependent constant, but it corresponds to the precise value in the $\zeta$-function scheme that we are using.

Note that the $\ell$-dependence in \eqref{zeromass1} also follows from Eq.\ \eqref{dimzerodet} at $M=0$. As a result, the conformal anomaly doesn't really provide any new information in the flat case. In this respect, the situation is more interesting in the curved cases. Evaluating the action $\tilde S_{\text L}$ for the conformal factors $\sigma_{\pm}$ given in \eqref{conffactor}, using in particular \eqref{ellminrel}, we get
\be\ln\mathscr D^{\pm}(0,\ell,L) = \ln D^{0}(0)\pm\frac{A^{\pm}}{3\pi L^{2}}-\frac{1}{3}\ln(2r_{0}L)\, ,\ee
where $r_0$ is related to $\ell$ by the relation given in \eqref{ellminrel}. The $L$-dependence is as predicted by \eqref{dimcurvdet} and we also obtain
\be\label{zeromasscurv} \ln D^{\pm}(0,r_{0}) = \ln D^{0}(0)\pm\frac{A^{\pm}}{3\pi}-\frac{1}{3}\ln(2r_{0})\, .\ee
%

%
%

\noindent\emph{Remark:} 

This formula predicts a nice large $\ell$ expansion in negative curvature, of the form
\be\label{largelzerom}\ln D^{-}\bigl(0,r_{0}(\ell)\bigr) = -\frac{\ell}{3\pi}+\frac{1}{4} - \frac{1}{2}\ln(2\pi)-2\zeta_{\text R}'(-1) +\frac{2\pi^{3}}{9\ell^{3}} - \frac{8\pi^{5}}{15\ell^{5}}+O\bigl(1/\ell^{7}\bigr)\, .\ee
As expected, the leading term is proportional to the area, which is $A^{-}= \ell + O(1)$ at leading order. However, the coefficient $-\frac{1}{3\pi}$ does not match with the coefficient predicted by the Plancherel measure, Eq.\ \eqref{dethyperinf}, which is approximately $0.0538$. This is a direct illustration of the fact, explained at the end of Section \ref{infSec}, that the large $\ell$ limit of the determinants do not match with the $\ell=\infty$ Plancherel description.

\section{\label{SLSec}Sturm-Liouville determinants}

Sturm-Liouville determinants have been considered by many authors in the literature \cite{SLdet} and the results we need and that we present below are along the same lines.

\begin{proposition}\label{SLdetProp} We have, for $n\geq 0$,
\begin{align}\label{ratioSLzero}&\frac{d^{0}_{n}(M^{2})}{d^{0}_{n}(0)} = \frac{\det(L_{n}^{0}+M^{2})}{\det L_{n}^{0}}=\prod_{k\geq 0}\biggl[1+\frac{M^{2}}{\la^{0}_{n,k}}\biggr]= \frac{2^{n}n!}{M^{n}}I_{n}(M)\, , \\
\label{ratioSLcurved}&\frac{d^{\pm}_{n}(M^{2})}{d^{\pm}_{n}(0)} = \frac{\det(L_{n}^{\pm}+M^{2})}{\det L_{n}^{\pm}}=\prod_{k\geq 0}\biggl[1+\frac{M^{2}}{\la^{\pm}_{n,k}}\biggr] =f_{n}^{\pm}(1;-M^{2})\, ,
\end{align}
where the functions $f_{n}^{\eta}$ are defined in \eqref{eigenflat} and \eqref{eigenminus}.
\end{proposition}
\begin{proof} Weyl's law, Eq.\ \eqref{Weyllaw}, implies that the series representation for the difference of the $\zeta$ functions
\be\label{zetadif} \zeta_{n}^{\eta}(s;M^{2}) - \zeta_{n}^{\eta}(s;0) = \sum_{k\geq 0}\biggl[\frac{1}{\bigl(\la^{\eta}_{n,k}+M^{2}\bigr)^{s}} - \frac{1}{\bigl(\la^{\eta}_{n,k}\bigr)^{s}}\biggr]\ee
converges for $\re s >-1/2$. The derivative with respect to $s$ at $s=0$ can thus be computed by taking the derivative of each term in the sum. Using \eqref{SLdetdef}, this yields the infinite product representations in \eqref{ratioSLzero} and \eqref{ratioSLcurved}. These infinite products can be computed using standard ideas from complex analy\-sis, as we now explain. 

In the zero curvature case, Eq.\ \eqref{eigenflat} and \eqref{Dirichleteigen} show that the eigenvalues $\la_{n,k}^{0}$ are the squares of the zeros of the Bessel function $J_{n}$. The result then follows from the standard infinite product representation of the Bessel function in terms of its zeros and from the identity $(iz)^{-n}J_{n}(iz) = z^{-n}I_{n}(z)$. 

The cases of non-zero curvature are conceptually similar. Let us first note that, by Weyl's law \eqref{Weyllaw}, the infinite products $\prod_{k\geq 0}(1-z/\la^{\pm}_{n,k})$ define entire functions of the complex variable $z$, with simple zeros at $z=\la^{\pm}_{n,k}$. The functions $f_{n}^{\pm}(1;z)$ are also entire. As a consequence of the Dirichlet condition \eqref{Dirichleteigen}, the zeros of these functions coincide with the eigenvalues $\la_{n,k}^{\pm}$; in particular, they are all real and positive. One can also check that these zeros are all simple. The ratio
\be\label{ratioSLfun} \frac{f_{n}^{\pm}(1;z)}{\prod_{k\geq 0}(1-z/\la^{\pm}_{n,k})}= e^{h(z)}\ee
is thus a non-vanishing entire function and can be written as the exponential of an entire function $h$. Consider then
\be\label{logder} h'(z) =\frac{\partial_{z}f_{n}^{\pm}}{f_{n}^{\pm}} - \sum_{k\geq 0}\frac{1}{z-\la^{\pm}_{n,k}}\, \cdotp\ee
If we could prove that this derivative vanishes, we would conclude that $h$ is a constant and, by evaluating the left-hand side of \eqref{ratioSLfun} at $z=0$ and using $f_{n}^{\pm}(1;0)=1$, we would find that this constant is zero, proving Eq.\ \eqref{ratioSLcurved}.

To derive that $h'=0$, we consider a closed contour $\mathcal C^{\pm}_{k}$ such that: i) it encircles the first $k$ eigenvalues $\smash{\la^{\pm}_{n,0},\ldots,\la^{\pm}_{n,k-1}}$; ii) it never comes close to any eigenvalue, i.e.\ the distance between $\mathcal C^{\pm}_{k}$ and any eigenvalue is always greater than a strictly positive, $k$-independent constant (this can always be achieved because the distance between successive eigenvalues increases with $k$ at large $k$ according to Weyl's law, Eq.\ \eqref{Weyllaw}); iii) $|z|\rightarrow\infty$ along the contour when $k\rightarrow\infty$. Concretely, we can choose $\mathcal C^{\pm}_{k}$ to be a circle centered at $z=0$ and of radius $\smash{R^{\pm}_{k}=\frac{1}{2}(\la^{\pm}_{n,k-1}+\la^{\pm}_{n,k})}$. Using Cauchy theorem, we note that
\be\label{intrdef} \frac{1}{2i\pi}\oint_{\mathcal C_{k}}\frac{\d z'}{z'-z}\frac{\partial_{z}f_{n}^{\pm}}{f_{n}^{\pm}}(1;z') = \frac{\partial_{z}f_{n}^{\pm}}{f_{n}^{\pm}}\bigl(1;z\bigr) - \sum_{k'=0}^{k-1}\frac{1}{z-\la^{\pm}_{n,k'}}\,\cvp\ee
for all $z$ in the interior of the circle $\mathcal C_{k}$. When $k\rightarrow\infty$, the right-hand side of this equation converges to $h'(z)$. To evaluate the left-hand side, we use an asymptotic analysis of $f_{n}^{\pm}$. The result we need is discussed in \cite{Jones} and can be cast in the form
\be\label{fnasymp} f_{n}^{\pm}(1;z) \underset{|z|\rightarrow\infty}{\sim} c^{\pm}_{n} z^{-n/2} J_{n}(a^{\pm}\sqrt{z})\, ,\ee
where the $a^{\pm}$ are defined in \eqref{aetadef} and the $c_{n}^{\pm}$ depend on $r_{0}$ and $n$ but not on $z$. Using the standard asymptotics for the Bessel function $J_{n}$, we get in this way
\begin{multline}\label{logfnasymp} \frac{\partial_{z}f_{n}^{\pm}}{f_{n}^{\pm}}(z) = 
-\frac{n}{2z}-\frac{a^{\pm}}{2\sqrt{z}}\Biggl[\frac{1}{2a^{\pm}\sqrt{z}} + \Bigl(1+ O\bigl(1/z\bigr)\Bigr)\tan\Bigl(a^{\pm}\sqrt{z}-\frac{n\pi}{2}-\frac{\pi}{4}\Bigr)\\
+ \frac{4n^{2}-1}{8a^{\pm}\sqrt{z}}\Bigl(1+ O\bigl(1/z\bigr)\Bigr)\cos^{-2}\Bigl(a^{\pm}\sqrt{z}-\frac{n\pi}{2}-\frac{\pi}{4}\Bigr)\Biggr]\, .\end{multline}
On the contours $\mathcal C^{\pm}_{k}$, we have $\smash{z = R_{k}^{\pm}e^{i\alpha}}$. When $|\alpha|>\varepsilon$, where $\varepsilon$ is a small, $k$-independent strictly positive constant, one deduces from \eqref{logfnasymp} and the elementary properties of the tangent and cosine functions for complex arguments with a non-zero imaginary part, that $\smash{|\partial_{z}f_{n}^{\pm}/f_{n}^{\pm}|}$ is bounded above by $\smash{C/\sqrt{R_{k}^{\pm}}}$, for some $k$-independent constant $C$. When $|\alpha|\leq\varepsilon$, this is also true, thanks to the fact that, by construction, the distance between the contour and the eigenvalues $\la_{n,k'}^{\pm}$, or equivalently the poles of the tangent and inverse cosine functions appearing in the right-hand side of Eq.\ \eqref{logfnasymp}, is always strictly greater than a fixed, $k$-independent constant. When $k\rightarrow\infty$, the integrand in the left-hand side of Eq.\ \eqref{intrdef} thus goes to zero as $\smash{1/(R_{k}^{\pm})^{3/2}}$ and the integral itself as $\smash{1/\sqrt{R_{k}^{\pm}}}$.
\end{proof}

\section{\label{ratioSec}Ratio of determinants}

\subsection{Naive discussion}

Naively, following the same logic as in Section \ref{SLSec}, and using \eqref{ratioSLzero}, we could be tempted to write, e.g.\ in the flat case,
\be\label{naiveratio} \frac{D^{0}(M^{2})}{D^{0}(0)} =\prod_{n\in\mathbb Z}\prod_{k\geq 0}\biggl[1+\frac{M^{2}}{\la^{0}_{n,k}}\biggr] = I_{0}(M)\prod_{n\geq 0}\biggl[\frac{2^{n}n!}{M^{n}}I_{n}(M)\biggr]^{2}\, .\ee
However, using the large $n$ asymptotics of $I_{n}(z)$,\footnote{This expansion is valid uniformly in $z$, as long as $|\arg z|<\frac{\pi}{2} - \epsilon$ for any small $\epsilon>0$ and $z$ is not greater than $O(n)$.}
\be\label{Inasymp1} I_{n}(z) = \frac{z^{n}}{\sqrt{2\pi}}\frac{e^{\sqrt{n^{2}+z^{2}}}}{(n^{2}+z^{2})^{1/4} (n+\sqrt{n^{2}+z^{2}})^{n}}\biggl[1-\frac{1}{12\sqrt{n^{2}+z^{2}}} + O\bigl(1/n^{2}\bigr)\biggr]\, ,\ee
we realize that the infinite product is divergent, since
\be\label{asympprod1} \ln\Bigl(\frac{2^{n}n!}{M^{n}}I_{n}(M)\Bigr) = \frac{M^{2}}{4 n} + O\bigl(1/n^{2}\bigr)\, .\ee
In contrast to one-dimensional Sturm-Liouville operators, taking the ratio of two-dimensional operators does not yield a finite answer. In particular, the product of the regularized Sturm-Liouville determinants does not yield the regularized two-dimensional determinant.

\subsection{Correct discussion}

The difficulty with the ratio of two-dimensional determinants comes from the fact that, when we consider the subtraction
\be\label{zetadif2d} \zeta^{\eta}(s;M^{2}) - \zeta^{\eta}(s;0) = \sum_{p\geq 0}\biggl[
\frac{1}{\bigl(\la_{p}^{\eta}+M^{2}\bigr)^{s}} - \frac{1}{\bigl(\la_{p}^{\eta}\bigr)^{s}}\biggr]\, ,\ee
Weyl's law, Eq.\ \eqref{Weyllaw}, shows that the resulting series representation converges only for $\re s>0$. Since $s=0$ is not in the interior of the convergence region, we cannot take the derivative under the summation and then set $s=0$; the resulting series would diverge. This is in contrast to what happens for one-dimensional determinants, see Eq.\ \eqref{zetadif} and below.

To cure the problem, we thus consider 
\be\label{diffzeta2} \sum_{p\geq 0}\biggl[\frac{1}{\bigl(\la_{p}^{\eta}+M^{2}\bigr)^{s}} - \frac{1}{\bigl(\la_{p}^{\eta}\bigr)^{s}}+\frac{sM^{2}}{\bigl(\la_{p}^{\eta}\bigr)^{s+1}}\biggr] = \zeta^{\eta}(s;M^{2})-\zeta^{\eta}(s;0)+ sM^{2}\zeta^{\eta}(s+1;0)\, ,\ee
which converges for $\re s>-1$. We can thus take the derivative of both sides of this equation with respect to $s$ and then take the limit $s\rightarrow 0$. To obtain the correct limit of the right-hand side, one must recall that the $\zeta$-functions $\zeta^{\eta}(s;M^{2})$ have a single pole at $s=1$, with residue given by Weyl's law,
\be\label{residue} \res_{s=1}\zeta^{\eta} = \frac{A^{\eta}}{4\pi}\,\cdotp\ee
This is a standard result that can be derived by using the heat kernel expansion, see \cite{Gilkey,bfref}. Introducing the so-called ``integrated two-point function at coinciding points''
\be\label{Cconstant} C^{\eta} = \lim_{s\rightarrow 1}\Bigl(\zeta^{\eta}(s;0) - \frac{A^{\eta}}{4\pi}\frac{1}{s-1}\Bigr)\ee
we get the correct infinite product formula for the ratio of two-dimensional determinants,
\be\label{ratioformula} \frac{D^{\eta}(M^{2})}{D^{\eta}(0)} = e^{C^{\eta} M^{2}}\prod_{p\geq 0}\biggl[\frac{\la_{p}^{\eta}+M^{2}}{\la_{p}^{\eta}}e^{-M^{2}/\la_{p}^{\eta}}\biggr]\, .\ee
Note that the right-hand side of this formula takes the form of a Weierstrass canonical product associated with the sequence $(\la_{p}^{\eta})_{p\geq 0}$. This is consistent with the general analysis presented in \cite{Voros}, where the relevance of the Weierstrass canonical product in $\zeta$-function regularisation was discussed.

\section{\label{flatSec}Determinants on the flat disk}

In the flat case, Eq.\ \eqref{ratioformula} reads
\be\label{ratflat} \frac{D^{0}(M^{2})}{D^{0}(0)} = e^{C^{0} M^{2}}\prod_{n\in\mathbb Z} \prod_{k\geq 0}\Biggl[\biggl(1+\frac{M^{2}}{\la^{0}_{n,k}}\biggr)e^{-M^{2}/\la^{0}_{n,k}}\Biggr]\, .\ee
The product over $k$ is evaluated as
\be\label{prodkeva} \prod_{k\geq 0}\Biggl[\biggl(1+\frac{M^{2}}{\la^{0}_{n,k}}\biggr)e^{-M^{2}/\la^{0}_{n,k}}\Biggr] =\Biggl[ \prod_{k\geq 0}\biggl(1+\frac{M^{2}}{\la^{0}_{n,k}}\biggr)\Biggl] e^{-M^{2}\sum_{k\geq 0}\frac{1}{\la^{0}_{n,k}}} =\frac{2^{n}n!}{M^{n}}I_{n}(M)e^{-M^{2}\zeta^{0}_{n}(1;0)}\, ,\ee
where we have used Eq.\ \eqref{ratioSLzero} and the definition of the Sturm-Liouville $\zeta$ function given in Eq.\ \eqref{zetaSL}. Overall we thus obtain
\be\label{resflat1} \frac{D^{0}(M^{2})}{D^{0}(0)}= e^{C^{0}M^{2}}I_{0}(M)e^{-M^{2}\zeta^{0}_{0}(1;0)}\prod_{n\geq 1}\biggl[\frac{2^{n}n!}{M^{n}}I_{n}(M)e^{-M^{2}\zeta^{0}_{n}(1;0)}\biggr]^{2}\, .\ee
The constants $\zeta^{0}_{n}(1;0)$ can be easily computed by expanding the second equality in Eq.\ \eqref{ratioSLzero} in power of $M^{2}$ and identifying the coefficients of $M^{2}$, using
\be\label{Besselseries} I_{n}(z) = (z/2)^{n}\sum_{k=0}^{\infty}\frac{(z/2)^{2k}}{k!(n+k)!}\,\cdotp\ee
One finds
\be\label{zetanone} \zeta^{0}_{n}(1;0) = \frac{1}{4(n+1)}\,\cdotp\ee

To compute $C^0$ requires more work. This constant is related to the divergence of the series representation
\be\label{zetazetan} \zeta^{0}(s;M^{2}) = \zeta^{0}_{0}(s;M^{2})+2\sum_{n\geq 1}\zeta^{0}_{n}(s;M^{2})\ee
when $s\rightarrow 1^{+}$ and thus to the large $n$ behaviour of $\zeta^{0}_{n}(s;M^{2})$. We are going to show that
\be\label{zetanlargen} \zeta^{0}_{n}(s;0)\underset{n\rightarrow\infty}{\sim}\frac{1}{4\sqrt{\pi}}\frac{\Gamma(s-1/2)}{s\Gamma(s)}\frac{1}{n^{2s-1}}\,\cdotp\ee
Assuming for the moment that this is correct, consider
\begin{multline}\label{sumzetansub}\sum_{n\geq 1}\biggl(\zeta^{0}_{n}(s;0) - \frac{1}{4\sqrt{\pi}}\frac{\Gamma(s-1/2)}{s\Gamma(s)}\frac{1}{n^{2s-1}}\biggr)
=\\ \frac{1}{2}\zeta^{0}(s;0) - \frac{1}{2}\zeta^{0}_{0}(s;0) 
- \frac{1}{4\sqrt{\pi}}\frac{\Gamma(s-1/2)}{s\Gamma(s)}\zeta_{\text R}(2s-1)\, .\end{multline}
The series on the left-hand side converges for $\re s>0$. Taking the $s\rightarrow 1$ limit on both sides of the equality yields
\be\label{Cvalue} C^0 = \frac{1}{2}\Bigl(\gamma -1 -\ln 2\Bigr)\, .\ee

To prove \eqref{zetanlargen}, we use the following contour integral representation
\be\label{contourzetan} \zeta^{0}_{n}(s;0)=\int_{\Gamma}\frac{\d \la}{2i\pi} \frac{r_{n}(\la)}{\la^{s}}\, \cvp\ee
which is valid when $\re s>1/2$, where the resolvent function is
\be\label{rndef} r_{n}(\la) = \sum_{k\geq 0}\frac{1}{\la-\la^{0}_{n,k}} = -\biggl[\frac{1}{2z}\frac{\d}{\d z}\ln\Bigl(\bigl(z/2\bigr)^{-n}I_{n}(z)\Bigr)\biggr]\bigl(z=i\sqrt{\la}\bigr)\, ,\ee
and the contour $\Gamma$ encircles counterclockwise all the $\la^{0}_{n,k}$ (we take it to be close to the positive real $\la$-axis, going from $+\infty$ to $\la^{0}_{n,0}-\epsilon$ with a small positive imaginary part and going back from $\smash{\la^{0}_{n,0}-\epsilon}$ to $+\infty$ with a small negative imaginary part). Using the large $|\la|$ behaviour of $I_{n}(i\sqrt{\la})=i^{n}J_{n}(\sqrt{\la})$ for $\epsilon<|\arg \la|<\pi-\epsilon$, it is straightforward to check that the contour can be deformed to go around the branch cut on the negative real $\la$-axis counterclockwise. With $v=-\la$, this yields
\be\label{contourzetan2} \zeta^{0}_{n}(s;0)=-\frac{\sin(\pi s)}{\pi}\int_{0}^{\infty}\!\d v\, \frac{r_{n}(-v)}{v^{s}}\ee
when $1/2<\re s<1$.\footnote{Here we have ignored the contribution of the small arc around the branch point which connects the two legs of the deformed contour. This is valid as long as $\re s<1$.}
Using the uniform asymptotics \eqref{Inasymp1}, we find, with $z=\sqrt{v}$,
\be\label{uniasused} \biggl[\frac{1}{2z}\frac{\d}{\d z}\ln\Bigl(\bigl(z/2\bigr)^{-n}I_{n}(z)\Bigr)\biggr]\bigl(z=\sqrt{v}\bigr)\underset{n\rightarrow\infty}{\sim} \frac{1}{2n}\frac{1}{1+\sqrt{1+v/n^{2}}}\,\cvp\ee
which, together with \eqref{rndef}, yields 
\be\label{zetanas3} \zeta^{0}_{n}(s;0)\underset{n\rightarrow\infty}{\sim}
\frac{\sin(\pi s)}{2\pi}\frac{1}{n^{2s-1}}\int_{0}^{\infty}\frac{\d x}{x^{s}(1+\sqrt{1+x})}\, \cdotp\ee
Performing the integral yields \eqref{zetanlargen}, a priori for $1/2<\re s<1$, but then for any $s$ by analytic continuation.

Let us note that Eq.\ \eqref{Cvalue} will be rederived in the next section from a more general method. 

Putting everything together, Eqs.\ \eqref{zeromassfinal},  \eqref{resflat1} and \eqref{Cvalue}, we get the explicit exact formula \eqref{detzero} for the determinant in flat space.

\noindent\emph{Remark:}

Combining \eqref{detzero} and \eqref{detexpan2} we get the large $z$ asymptotic expansion
\begin{multline}\label{largezasderived} \ln I_{0}(z) + 2\sum_{n=1}^{\infty}\ln\Bigl[\bigl(z/2\bigr)^{-n}n!I_{n}(z)e^{-\frac{z^{2}}{4(n+1)}}\Bigr] = \\
\frac{z^{2}}{2}\bigl(2+\ln 2 -\gamma-\ln z \bigr)-\frac{\pi z}{2}+\frac{1}{3}\ln z \\ 
-\frac{1}{3}\ln 2 + \frac{1}{2}\ln(2\pi) + \frac{5}{12}+2\zeta_{\text R}'(-1) - \frac{\pi}{128 z} + O\bigl(1/z^{2}\bigr)\, .
\end{multline}
Higher order terms could be obtained straightforwardly from the heat kernel expansion. It seems rather non-trivial to derive this expansion directly from the series representation on the left-hand side of the above equation, because we do not know an approximation of $I_{n}(z)$ at large $z$ uniformly in $n$ (whereas \eqref{Inasymp1} is an approximation of $I_{n}(z)$ at large $n$ uniformly in $z$). 

In the case of negative curvature, for which a similar formula for the determinant, Eq.\ \eqref{deteta}, will be derived below, the heat kernel method is useless to derive the interesting large $\ell$ asymptotics. It seems that this can be done in full generality only if a method to treat directly the asymptotics of the infinite series is devised. See Section \ref{specialSec2} for a further discussion in the special cases $M^{2}=q(q+1)$, $q\in\mathbb N$.

\section{\label{curveSec}Determinants on the curved disks}

\subsection{General analysis}

Starting from Eq.\ \eqref{ratioformula} and using \eqref{ratioSLcurved}, we get
\be\label{reshyper1}\frac{D^{\pm}(M^{2},r_{0})}{D^{\pm}(0,r_{0})} = e^{C^{\pm}M^{2}}f^{\pm}_{0}(1,-M^{2})e^{-M^{2}\zeta^{\pm}_{0}(1;0)}\prod_{n\geq 1}\biggl[f^{\pm}_{n}(1;-M^{2})e^{-M^{2}\zeta^{\pm}_{n}(1;0)}\biggr]^{2}\, ,\ee
which is the curved space version of \eqref{resflat1}.

The constants $\zeta^{\pm}_{n}(1;0)$ are computed by expanding the second equality in Eq.\ \eqref{ratioSLcurved} in power of $M^{2}$ and identifying the coefficients of $M^{2}$. This yields
\be\label{zetahtwo} \zeta_{n}^{\pm}(1;0) = -\partial_{\lambda}f_{n}^{\pm}(1;\la=0) =\mp (\partial_{1}-\partial_{2})F\bigl(1,0,n+1,\pm A^{\pm}/(4\pi)\bigr)\, ,\ee
where $\partial_{1}$ and $\partial_{2}$ are the partial derivatives with respect to the first and second argument of the hypergeometric function, respectively. From the usual series representation of the hypergeometric function, we get $\partial_{1}F(1,0,n+1,z)=0$ and
\be\label{hypergeomcalc}\partial_{2}F(1,0,n+1,z) = \frac{z}{n+1}F(1,1,n+2,z)\, ,\ee
from which we obtain
\be\label{zetanhypone} \zeta^{\pm}_{n}(1;0) = \frac{A^{\pm}}{4\pi}\frac{1}{n+1}F\bigl(1,1,n+2,\pm A^{\pm}/(4\pi)\bigr)\, .\ee

As in flat space, the computation of the constant $C^{\pm}$, defined in Eq.\ \eqref{Cconstant}, is the trickier part. We could try to follow the direct approach used in Section \ref{flatSec}. This would require to start by generalizing Eq.\ \eqref{zetanlargen}. Instead, we prefer to use a more general approach, which will also allow to find again the result \eqref{Cvalue}.

We are going to use the fact that a nice transformation law exists for $C$ as a function of the metric, when one performs a Weyl rescaling. This is similar to the conformal anomaly that we have used in Section \ref{zeromassSec} to compute the determinants at zero mass. We follow some of the reasonings in \cite{FKZ}; the case of the disk topology we are interested in is actually simpler than the closed compact case studied in this reference.\footnote{The reason is that on a closed compact surface there is a zero mode that yields non-trivial contributions. This zero mode is absent on a disk with Dirichlet boundary condition.}

We consider a disk with an arbitrary metric $g$. We note as usual the eigenvalues of the Laplacian with Dirichlet boundary conditions $\la_{p}$, with $0<\la_{0}<\la_{1}<\cdots$, and the corresponding eigenfunctions $\psi_{p}$. The space-dependent $\zeta$ function is defined by
\be\label{zetatotalgen} \zeta(s,x,y;g) = \sum_{p\geq 0}\frac{\psi_{p}(x)\psi_{p}(y)}{\la_{p}^{s}}\, \cvp\ee
generalizing \eqref{zetadef}. The two-point function is
\be\label{Gdef} G(x,y;g) = \zeta(1,x,y;g)\, .\ee
Note that these formulas have obvious generalizations to $M\not = 0$, but, for the present discussion, whose aim is to compute $C$, we need only the case of zero mass. 

Using standard perturbation theory to compute the infinitesimal variations of the eigenvalues and of the eigenfunctions under an infinitesimal Weyl rescaling of the metric, $\delta g = 2\delta\sigma g$, under which the Laplacian varies as $\delta\Delta = -2\delta\sigma \Delta$, one finds that the two-point function is invariant,
\be\label{GWeylinv} G\bigl(x,y;e^{2\sigma}g_{0}\bigr) =G\bigl(x,y;g_{0}\bigr) \, .\ee
The two-point function at coinciding point may be defined in two natural ways. One is to subtract directly the logarithmic divergence,
\be\label{Gren1} G_{\text R}(x;g) = \lim_{y\rightarrow x}\Bigl( G(x,y;g) + \frac{1}{2\pi}\ln d(x,y;g)\Bigr)\, ,\ee
where $d(x,y)$ denotes the geodesic distance, and another is to subtract instead the pole at $s=1$ in the zeta function,
\be\label{Grenzeta} G_{\text R}^{(\zeta)}(x;g) = \lim_{s\rightarrow 1}\Bigl(\zeta(s,x,x;g)-\frac{1}{4\pi}\frac{1}{s-1}\Bigr)\, .\ee
One can show (see e.g.\ \cite{FKZ}) that
\be\label{GRrel} G_{\text R}^{(\zeta)}(x;g) = G_{\text R}(x;g) + \frac{\gamma-\ln 2}{2\pi}\,\cdotp\ee
The function $G_{\text R}^{(\zeta)}$ (or equivalently $G_{\text R}$) thus transforms non-trivially under a Weyl rescaling, but all the dependence in the conformal factor $\sigma$ comes from the geodesic distance,
\be\label{GRtrans} G_{\text R}^{(\zeta)}\bigl(x;e^{2\sigma}g_{0}\bigr) = G_{\text R}^{(\zeta)}\bigl(x;g_{0}\bigr) + \frac{\sigma(x)}{2\pi}\,\cdotp\ee
The constant $C[g]$ defined by \eqref{Cconstant} is expressed as
\be\label{Cofgint} C[g] = \int\!\d^{2}x\sqrt{g}\, G_{\text R}^{(\zeta)}(x;g)\ee
and thus, using \eqref{GRtrans}, is also given by
\be\label{Cofgtrans} C[g]=\int\!\d^{2}x\sqrt{g_{0}}\,e^{2\sigma}\Bigl[ G_{\text R}^{(\zeta)}(x;g_{0})+\frac{\sigma(x)}{2\pi}\Bigr]\, .\ee

For instance, if $g_{0}=\delta$ is the metric for the flat disk of circumference $2\pi$ and $g=(\frac{\ell}{2\pi})^{2} \delta$ is the metric for the flat disk of circumference $\ell$, \eqref{Cofgtrans} yields
\be\label{exC} C\Bigl[\Bigl(\frac{\ell}{2\pi}\Bigr)^{2} \delta\Bigr] = \frac{A^{0}}{\pi}\biggl(  C\bigl[\delta\bigr]+\frac{1}{2}\ln\frac{\ell}{2\pi}\biggr)\, ,\ee
which agrees with the formula \eqref{Cvalue} for
\be\label{Cdelta} C[\delta] = \frac{1}{2}\bigl( \gamma - 1 -\ln 2\bigr)\, .\ee

We want to use \eqref{Cofgtrans} to compute $C[g]$ for the hyperbolic and spherical disks. We thus choose $g_{0}=\delta$ and the conformal factors as in Eq.\ \eqref{conffactor}. To proceed, we need $G_{\text R}^{(\zeta)}(x;\delta)$. The two-point function on the unit disk can be computed by the method of images and is given by
\be\label{Gonunitd} G(\vec x,\vec y;\delta) = -\frac{1}{2\pi}\ln\frac{|\vec x - \vec y|}{\sqrt{|\vec x|^{2}|\vec y|^{2}-2\vec x\cdot\vec y + 1}}\,\cdotp
\ee
It is straightforward to check that this formula satisfies all the requirements characterizing the two-point function: $G(\vec x,\vec y) = G(\vec y,\vec x)$, $\Delta_{\vec x}G(\vec x,\vec y) = \delta (\vec x-\vec y)$ for all $\vec x $ and $\vec y$ in the interior of the disk and $G(\vec x,\vec y)=0$ is $\vec x$ or $\vec y$ is on the boundary, i.e.\ $|\vec x|=1$ or $|\vec y|=1$. From \eqref{Gren1} and \eqref{GRrel} we thus get
\be\label{GRunitdisk} G_{\text R}^{(\zeta)}(\vec x;\delta) = \frac{1}{2\pi}\Bigl[\ln\bigl(1-r^{2}\bigr)+\gamma-\ln 2\Bigr]\, .\ee
As a check, using \eqref{Cofgint}, we get
\be\label{Cdeldirect} C[\delta] =\int_{0}^{1}\!\d r \, r \Bigl[\ln\bigl(1-r^{2}\bigr)+\gamma-\ln 2\Bigr] = \frac{1}{2}\bigl(\gamma-1-\ln2\bigr)\, ,\ee
which is consistent with \eqref{Cdelta}. Finally, Eq.\ \eqref{Cofgtrans} yields the constants $C=C^{\pm}$ for the spherical and hyperbolic disks,
\begin{align} C^{\pm} &= \int_{0}^{1}\!\d r \, \frac{4r_{0}^{2}r}{(1\pm r_{0}^{2}r^{2})^{2}}\Bigl[\ln\bigl(1-r^{2}\bigr)+\gamma-\ln 2+\ln\frac{2r_{0}}{1\pm r_{0}^{2}r^{2}}\Bigr]\\\label{Cfinal} &=  \frac{A^{\pm}}{2\pi}\Bigl(\gamma -1 + \ln r_{0}\Bigr)\, . 
\end{align}

Putting together \eqref{zeromassfinal}, \eqref{zeromasscurv}, \eqref{reshyper1}, \eqref{zetanhypone} and \eqref{Cfinal}, we find Eq.\ \eqref{deteta}.

\subsection{Special case: determinant on a hemisphere}

We will now use the formula for the determinant on a disk with constant positive curvature to obtain an exact expression for the determinant on a hemisphere\footnote{We thank the referee who asked us to consider this special case.}. In this special case, the boundary of the disk is a geodesic, and consequently we can make contact with the results derived previously in the literature. In particular, our result will match exactly with the form of the determinant that can be obtained from the  computations done in \cite{KMW}.

When the positive curvature disk corresponds to a hemisphere of unit radius, we have $r_0=1$. The area of the disk is then 
$A^+(1)=2\pi$. Moreover, the hypergeometric functions that appear in the expression of the determinant reduce to 
\begin{equation}
\begin{split}
f_n^+(1;-M^2)
&=F\Big(\frac{1}{2}(1+\sqrt{1-4M^2}),\frac{1}{2}(1-\sqrt{1-4M^2}),|n|+1,\frac{1}{2}\Big)\\
&=\frac{2^{-|n|}\sqrt{\pi}\Gamma(|n|+1)}{\Gamma\Big(\frac{1}{4}(2|n|+3+\sqrt{1-4M^2})\Big)\Gamma\Big(\frac{1}{4}(2|n|+3-\sqrt{1-4M^2})\Big)},
\end{split}
\end{equation}
and
\begin{equation}
\begin{split}
F\Big(1,1,|n|+2,\frac{1}{2}\Big)
&=(|n|+1)\biggl[\psi\Big(\frac{|n|+2}{2}\Big)-\psi\Big(\frac{|n|+1}{2}\Big)\biggr],
\end{split}
\end{equation}
where $\psi(z)=\Gamma'(z)/\Gamma(z)$ is the digamma function. Using these expressions, and replacing the Gamma function appearing in the formula by a ratio of Barnes $\mathsf{G}$-functions according to $\Gamma(z)=\mathsf{G}(z+1)/\mathsf{G}(z)$, we get the following form for the infinite series appearing in the logarithm of the determinant (see \eqref{deteta}),
\begin{equation}\label{infiniteseries:hemisphere}
\begin{split}
&\sum_{n\in\mathbb Z}\ln\Bigg[f_n^+(1;-M^2)e^{-\frac{M^2}{2(|n|+1)}F\Big(1,1,|n|+2,\frac{1}{2}\Big)}\Bigg]\\
&=\frac{1}{2}\ln\pi+\ln\Bigg[\frac{\mathsf{G}\Big(\frac{1}{4}(3+\sqrt{1-4M^2})\Big)\mathsf{G}\Big(\frac{1}{4}(3-\sqrt{1-4M^2})\Big)}{\mathsf{G}\Big(\frac{1}{4}(7+\sqrt{1-4M^2})\Big)\mathsf{G}\Big(\frac{1}{4}(7-\sqrt{1-4M^2})\Big)}\Bigg]-\frac{M^2}{2}\Big[\psi(1)-\psi\Big(\frac{1}{2}\Big)\Big]\\
&\quad+2\lim_{N\rightarrow\infty}\sum_{n=1}^N\Bigg\{\ln\Bigg[\frac{\mathsf{G}\Big(\frac{n}{2}+\frac{1}{4}(3+\sqrt{1-4M^2})\Big)\mathsf{G}\Big(\frac{n}{2}+\frac{1}{4}(3-\sqrt{1-4M^2})\Big)}{\mathsf{G}\Big(\frac{n+2}{2}+\frac{1}{4}(3+\sqrt{1-4M^2})\Big)\mathsf{G}\Big(\frac{n+2}{2}+\frac{1}{4}(3-\sqrt{1-4M^2})\Big)}\Bigg]\\
&\qquad\qquad\qquad-\frac{M^2}{2}\Big[\psi\Big(\frac{n+2}{2}\Big)-\psi\Big(\frac{n+1}{2}\Big)\Big]-n\ln 2+\frac{1}{2}\ln \pi+\ln \Big[\frac{\mathsf{G}(n+2)}{\mathsf{G}(n+1)}\Big]\Bigg\}.
\end{split}
\end{equation}
Note that the mutual cancellation of the terms in the finite sum above gives
\begin{equation}
\begin{split}
&\sum_{n=1}^N\Bigg\{\ln\Bigg[\frac{\mathsf{G}\Big(\frac{n}{2}+\frac{1}{4}(3+\sqrt{1-4M^2})\Big)\mathsf{G}\Big(\frac{n}{2}+\frac{1}{4}(3-\sqrt{1-4M^2})\Big)}{\mathsf{G}\Big(\frac{n+2}{2}+\frac{1}{4}(3+\sqrt{1-4M^2})\Big)\mathsf{G}\Big(\frac{n+2}{2}+\frac{1}{4}(3-\sqrt{1-4M^2})\Big)}\Bigg]\\
&\qquad-\frac{M^2}{2}\Big[\psi\Big(\frac{n+2}{2}\Big)-\psi\Big(\frac{n+1}{2}\Big)\Big]-n\ln 2+\frac{1}{2}\ln \pi+\ln \Big[\frac{\mathsf{G}(n+2)}{\mathsf{G}(n+1)}\Big]\Bigg\}\\
&=-\frac{N(N+1)}{2}\ln 2+\frac{N}{2}\ln\pi+\ln \Big(\mathsf{G}(N+2)\Big)-\frac{M^2}{2}\Big[\psi\Big(\frac{N+2}{2}\Big)-\psi(1)\Big]\\
&\quad-\ln\Bigg[\frac{\mathsf{G}\Big(\frac{N+2}{2}+\frac{1}{4}(3+\sqrt{1-4M^2})\Big)\mathsf{G}\Big(\frac{N+1}{2}+\frac{1}{4}(3+\sqrt{1-4M^2})\Big)}{\mathsf{G}\Big(\frac{1}{2}+\frac{1}{4}(3+\sqrt{1-4M^2})\Big)\mathsf{G}\Big(1+\frac{1}{4}(3+\sqrt{1-4M^2})\Big)}\Bigg]\\
&\quad-\ln\Bigg[\frac{\mathsf{G}\Big(\frac{N+2}{2}+\frac{1}{4}(3-\sqrt{1-4M^2})\Big)\mathsf{G}\Big(\frac{N+1}{2}+\frac{1}{4}(3-\sqrt{1-4M^2})\Big)}{\mathsf{G}\Big(\frac{1}{2}+\frac{1}{4}(3-\sqrt{1-4M^2})\Big)\mathsf{G}\Big(1+\frac{1}{4}(3-\sqrt{1-4M^2})\Big)}\Bigg].
\end{split}
\end{equation}
Here we have used the fact that $\mathsf{G}(1)=1$. To evaluate the $N\rightarrow\infty$ limit of the above expression, we need to use the following asymptotic forms of the functions $\psi(z)$ and $\mathsf{G}(z)$ in the large $z$ regime:
\begin{equation}
\begin{split}
&\psi(z)= \ln z+O(1/z),\\
&\mathsf{G}(z+1)= \frac{z^2}{2}\ln z-\frac{3z^2}{4}+\frac{z}{2}\ln(2\pi)-\frac{1}{12}\ln z+\frac{1}{12}-\ln \mathsf{A}+O(1/z^2),
\end{split}
\end{equation}
where $\mathsf{A}$ is the Glaisher-Kinkelin constant. Using these asymptotic forms, one can check that all the divergent pieces in the $N\rightarrow\infty$ limit cancel each other, as they should, and the series in \eqref{infiniteseries:hemisphere} reduces to
\begin{equation}
\begin{split}
&\sum_{n\in\mathbb Z}\ln\Bigg[f_n^+(1;-M^2)e^{-\frac{M^2}{2(|n|+1)}F\Big(1,1,|n|+2,\frac{1}{2}\Big)}\Bigg]\\
&=-\frac{1}{2}\ln\pi-\frac{1}{6}\ln 2+6\ln \mathsf{A}-\frac{1}{2}+\frac{M^2}{2}\Big[\psi(1)+\psi\Big(\frac{1}{2}\Big)\Big]\\
&\quad+\ln\Bigg[\mathsf{G}\Big(\frac{1}{4}(3+\sqrt{1-4M^2})\Big)\mathsf{G}\Big(\frac{1}{4}(5+\sqrt{1-4M^2})\Big)^2 \mathsf{G}\Big(\frac{1}{4}(7+\sqrt{1-4M^2})\Big)\Bigg]\\
&\quad+\ln\Bigg[\mathsf{G}\Big(\frac{1}{4}(3-\sqrt{1-4M^2})\Big)\mathsf{G}\Big(\frac{1}{4}(5-\sqrt{1-4M^2})\Big)^2 \mathsf{G}\Big(\frac{1}{4}(7-\sqrt{1-4M^2})\Big)\Bigg].
\end{split}
\end{equation}
This can be further simplified by using the following duplication formula for the Barnes $\mathsf{G}$-function,
\begin{multline}
\ln\Big[ \mathsf{G}(x)\mathsf{G}(x+\frac{1}{2})^2\mathsf{G}(x+1)\Big]\\
=\frac{1}{4}-3\ln \mathsf{A}+\Big(-2x^2+3x-\frac{11}{12}\Big)\ln 2+\Big(x-\frac{1}{2}\Big)\ln\pi+\ln \mathsf{G}(2x)\, ,
\end{multline}
which gives
\begin{multline}
\sum_{n\in\mathbb Z}\ln\Bigg[f_n^+(1;-M^2)e^{-\frac{M^2}{2(|n|+1)}F\Big(1,1,|n|+2,\frac{1}{2}\Big)}\Bigg]\\
=-\gamma M^2+\ln\Bigg[\mathsf{G}\Big(\frac{1}{2}(3+\sqrt{1-4M^2})\Big)\mathsf{G}\Big(\frac{1}{2}(3-\sqrt{1-4M^2})\Big)\Bigg].
\end{multline}
In deriving the coefficient of $M^2$, we have used the fact that 
\be
\psi(1)=-\gamma\, ,\quad \psi(\frac{1}{2})=-\gamma-2\ln(2)\, .
\ee

Plugging the above expression of the infinite series into the determinant given in \eqref{deteta}, we get
\begin{multline}
\ln D^+(M^2,1)
=-\frac{1}{2}\ln(2\pi)+\frac{1}{4}-2\zeta_{\text R}'(-1)- M^2 \\
\quad+\ln\Bigg[\mathsf{G}\Big(\frac{1}{2}(3+\sqrt{1-4M^2})\Big) \mathsf{G}\Big(\frac{1}{2}(3-\sqrt{1-4M^2})\Big)\Bigg].
\end{multline}
Finally, in order to make contact with the result derived in \cite{KMW}, we use the relation $\mathsf{G}(z+1)=\Gamma(z) \mathsf{G}(z)$ and Euler's reflection formula,
\begin{equation}
\begin{split}
\Gamma(z)\Gamma(1-z)=\frac{\pi}{\sin(\pi z)}=\frac{\pi}{\cos(\frac{\pi}{2}-\pi z)}\,\cvp
\end{split}
\end{equation}
to obtain
\begin{multline}\label{dethemisphere:final result}
\ln D^+(M^2,1)
=-\frac{1}{2}\ln(2\pi)+\frac{1}{4}-2\zeta_{\text R}'(-1)- M^2-\ln\Bigg[\frac{1}{\pi}\cos\Big(\frac{\pi}{2}(\sqrt{1-4M^2}) \Big)\Bigg]\\
\quad+\ln\Bigg[\mathsf{G}\Big(\frac{1}{2}(1+\sqrt{1-4M^2})\Big) \mathsf{G}\Big(\frac{1}{2}(1-\sqrt{1-4M^2})\Big)\Bigg].
\end{multline}
This is exactly the result that one would obtain from the computations done in \cite{KMW}.\footnote{In \cite{KMW}, the authors compute separately the determinant on the sphere and its ratio with the square of the determinant on the hemisphere. The result \eqref{dethemisphere:final result} can be obtained by considering the ratio of these two quantities. }

\section{\label{specialSec}The case $M^{2}=-\eta q(q+1)$ for $q\in \mathbb{N}$}

We now focus on the special case where the mass parameter is given either by $M^{2}=M^{2}_{-,q} = + q(q+1)$, for negative curvature, or by $M^{2}=M^{2}_{+,q} = - q(q+1)$, for positive curvature, as in Eq.\ \eqref{specialmasses}, for any positive integer $q$. In these cases, we shall be able to evaluate much more explicitly the infinite sum appearing in our formula \eqref{deteta} for the determinants. The simplification comes from the fact that the hypergeometric functions $f_{n}^{\eta}$ appearing in \eqref{deteta} can be expressed in terms of polynomials for these cases, see e.g.\ \eqref{fnoteval} and \eqref{fneval}. The case $q=1$ is crucially relevant for the quantum gravity calculations presented in \cite{Loopcalc} and the result in this case is given in Eq.\ \eqref{detspecial}. 

\subsection{General formula}

We note $P_{q}$ the Legendre polynomial of degree $q$ and $\omega_{q,k}^{\eta}(r_{0})$, $1\leq k\leq q$, the roots of the Meixner polynomial $M_{q}(x;-2q,-\eta r_{0}^{2})$. We recall that the Meixner polynomial is a degree $q$ polynomial that can be expressed in terms of the hypergeometric function as
\be\label{Meixnerdef} M_{q}(x;b,c) = (b)_{q}F(-q,-x,b,1-1/c)\, ,\ee
where $(b)_{q} = b (b+1)\cdots (b+q-1)$ is the Pochhammer symbol. In our case, we have
\be\label{Meixnerid}\begin{split} M_{q}(x;-2q,-\eta r_{0}^{2}) & = (-1)^{q}\frac{(2q)!}{q!}F\bigl(-q,-x,-2q,1+\eta/r_{0}^{2}\bigr)\\ & = \Bigl(\frac{\eta r_{0}^{2}}{1+\eta r_{0}^{2}}\Bigr)^{-q}(x+1-q)_{q}F\Bigl(-q,1+q,x+1-q,\frac{\eta r_{0}^{2}}{1+\eta r_{0}^{2}}\Bigr)\\
& = \Bigl(\frac{\eta r_{0}^{2}}{1+\eta r_{0}^{2}}\Bigr)^{-q}\prod_{k=1}^{q}\bigl(x-\omega_{q,k}^{\eta}(r_{0})\bigr)\, ,
\end{split}\ee
where the second equality can be easily checked by using the explicit series representation of the hypergeometric function. The main result of this section is as follows.
\begin{proposition} The determinants given by Eq.\ \eqref{deteta} are expressed as
\begin{multline}\label{specialdetform}
\ln D^{\eta}\bigl(-\eta q(q+1),r_{0}\bigr) = - \frac{1}{2}\ln(2\pi)-\frac{5}{12}-2\zeta_{\text R}'(-1)+\frac{\eta A^{\eta}(r_{0})}{3\pi} - \frac{1}{3}\ln r_{0}\\
-\frac{\eta q(q+1)}{2\pi}\bigl(\ln r_{0}-1\bigr)  A^{\eta}(r_{0})
-q(q+1)\frac{1-\eta r_{0}^{2}}{1+\eta r_{0}^{2}}\ln\bigl(1+\eta r_{0}^{2}\bigr)\\
+\ln P_q\Bigl(\frac{1-\eta r_0^2}{1+\eta r_0^2 }\Bigr)-2\ln\prod_{k=1}^{q}\frac{\Gamma\bigl(1+q-\omega_{q,k}^{\eta}(r_{0})\bigr)}{k!}\, .
\end{multline}
\end{proposition}
For $q=1$, one can check straightforwardly that we find the result announced in Eq.\ \eqref{detspecial}. Note that, for negative curvature, the Legendre polynomial and the product of the $\Gamma$ functions in the argument of the logarithm are always strictly positive real numbers, ensuring that $\ln D^{-}$ is always real. This is actually more generally valid for any positive value of $M^2$ due to the positivity of the Laplacian $\Delta_-$  (see appendix \ref{AppA}). For $M^2=M_{-,q}^2$ we will  make this explicit below. In positive curvature, the determinant may be a negative real number. This is related to the fact that the operator whose determinant we compute may then have negative eigenvalues\footnote{Note that $M_{+,q}^2=-q(q+1)<0$. So, although the Laplacian $\Delta_+$ is a positive operator (see appendix \ref{AppA}), the operator $(\Delta_++M_{+,q}^2)$ need not be positive.}, a case which is dealt with as mentioned at the end of Section \ref{SLoneSec}.

\begin{proof} In order to simplify some formulas, we introduce the notations
\be\label{zetattdef} z_{\eta} = \frac{\eta A^{\eta}}{4\pi} = \frac{\eta r_{0}^{2}}{1+\eta r_{0}^{2}}\, \cvp\quad \tilde z_{\eta} = \frac{z_{\eta}-1}{z_{\eta}} = -\frac{\eta}{r_{0}^{2}}\ee
for $\eta=\pm$. For later reference, we note that $z_{+}\in [0,1[$, or equivalently $\tilde z_{+}<0$; and that $z_{-}\leq 0$, or equivalently $\tilde z_{-}>1$.

Let us start by evaluating the term $n=0$ in the infinite sum over $n\in\mathbb Z$ in Eq.\ \eqref{deteta}. We have
\be\label{fnoteval} f_{0}^{\eta}(1;-M^{2}_{\eta,q}) = F(-q,q+1,1,z_{\eta}) = \frac{z_{\eta}^{q}}{q!}M_{q}\Bigl(q;-2q,\frac{z_{\eta}}{z_{\eta}-1}\Bigr)=P_{q}(1-2z_{\eta})\, ,\ee
where $P_{q}$ is the Legendre polynomial of degree $q$. Note that this is always strictly positive in negative curvature, as can be checked straightforwardly from the expansion of the hypergeometric function and the fact that $z_{-}\leq 0$. This produces the term $\ln P_{q}$ in Eq.\ \eqref{specialdetform}. Moreover, 
\be\label{Fid1} F(1,1,2,z) = -\frac{\ln (1-z)}{z}\, \cdotp\ee
The $n=0$ term is thus
\be\label{nnotterm} \ln P_{q}\Bigl(\frac{1-\eta r_{0}^{2}}{1+\eta r_{0}^{2}}\Bigr) + q(q+1)\ln (1+\eta r_{0}^{2}) = \ln P_{q}(1-2 z_{\eta}) - q(q+1)\ln (1-z_{\eta})\, .\ee

Using this, noting that the contribution of the modes is symmetric under $n\rightarrow -n$ and comparing Eqs.\ \eqref{deteta} and \eqref{specialdetform}, we see that the result we want to prove is equivalent to the identity
\begin{multline}\label{toprove1} \sum_{n\geq 1}\ln\biggl[f_{n}^{\eta}(1;-M^{2}_{\eta,q})
e^{-\frac{A^{\eta}M^{2}_{\eta,q}}{4\pi(n+1)}F(1,1,n+2,\frac{\eta A^{\eta}}{4\pi})}\biggr] =\\
q(q+1)\bigl(\gamma z_{\eta}+(1-z_{\eta}) \ln (1-z_{\eta})\bigr) -\ln\prod_{k=1}^{q}\frac{\Gamma\bigl(1+q-\omega_{q,k}^{\eta}(r_{0})\bigr)}{k!}\, \cdotp
\end{multline}
Our strategy to evaluate the sum on the left-hand side of this equation is to write it as the limit of finite sums,
\be\label{stra1} \sum_{n\geq 1}\ln\biggl[f_{n}^{\eta}(1;-M^{2}_{\eta,q})
e^{-\frac{A^{\eta}M^{2}_{\eta,q}}{4\pi(n+1)}F(1,1,n+2,\frac{\eta A^{\eta}}{4\pi})}\biggr] =\lim_{N\rightarrow\infty} \Bigl(q(q+1)\Sigma_{N}^{(1)} + \Sigma_{N}^{(2)}\Bigr)\, ,\ee
where, using, in particular, Eq.\ \eqref{zetattdef}, 
\begin{align}\label{Sigsum1} \Sigma_{N}^{(1)} &= \sum_{n=1}^{N}\frac{z_{\eta}}{n+1}F(1,1,n+2,z_{\eta})
\\
\label{Sigsum2}\Sigma_{N}^{(2)} &=\sum_{n=1}^{N}\ln f_{n}^{\eta}(1;-M^{2}_{\eta,q})\, .
\end{align}
\begin{lemma}\label{lemmaone} We have, when $N\rightarrow\infty$,
\be\label{SigNoneform} \Sigma_{N}^{(1)} = z_{\eta}\ln N + \gamma z_{\eta} +(1-z_{\eta})\ln (1-z_{\eta}) + O(1/N)\, .\ee
\end{lemma}
\begin{proof}
To prove the lemma, we start by noting that, for all $z<1$, 
\be\label{Fidn} F(1,1,n+2,z)=\frac{1+n}{z}\Biggl[\sum_{s=1}^{n}\frac{1}{s}\Big(\frac{z-1}{z}\Big)^{n-s}-\Big(\frac{z-1}{z}\Big)^{n}\ln(1-z)\Biggr],\ee
generalizing \eqref{Fid1}. We can use this formula for our purposes, because the variable  $z_{\eta}$ defined in Eq.\ \eqref{zetattdef} is always strictly less than one. Performing explicitly the trivial geometric sum, we get
\be\label{Sigo1} \Sigma_{N}^{(1)} = \sum_{n=1}^{N}\sum_{s=1}^{n}\frac{\tilde z_{\eta}^{n-s}}{s} +(1-\tilde z_{\eta}^{N})(1-z_{\eta})\ln (1-z_{\eta})\, .\ee
By reordering the terms and evaluating one more geometric sum, the first sum in the above equation can be written as
\be\label{Sif01fir} \sum_{n=1}^{N}\sum_{s=1}^{n}\frac{\tilde z_{\eta}^{n-s}}{s} = 
\sum_{s=1}^{N}\frac{1}{s}\sum_{p=0}^{N-s}\tilde z_{\eta}^{p} = z_{\eta} H_{N} +(1- z_{\eta})\tilde z_{\eta}^{N}\sum_{s=1}^{N}\frac{\tilde z_{\eta}^{-s}}{s}\,\cvp\ee
where the
\be\label{HNdef} H_{N} = \sum_{s=1}^{N}\frac{1}{s} = \ln N + \gamma + O(1/N)\ee
are the harmonic numbers. Putting Eqs.\ \eqref{Sigo1} and \eqref{Sif01fir} together yields
\be\label{Sigo2} \Sigma_{N}^{(1)} = z_{\eta}H_{N} + (1-z_{\eta})\ln (1-z_{\eta}) + (1-z_{\eta})R_{N}\ee
for a remainder
\be\label{remainder} R_{N} = \tilde z_{\eta}^{N}\biggl(\sum_{s=1}^{N}\frac{\tilde z_{\eta}^{-s}}{s} - \ln(1-z_{\eta})\biggr)\, .\ee
Taking into account the form of the large $N$ expansion of the harmonic numbers given in Eq.\ \eqref{HNdef}, the equation \eqref{SigNoneform} we want to prove is equivalent to the statement that
\be\label{RNexp} R_{N} = O\bigl(1/N\bigr)\ee
at large $N$. To prove this asymptotic behaviour, we distinguish three cases.

\paragraph{Case one: $|\tilde z_{\eta}|>1$.} Note that this condition is always true in negative curvature, but not necessarily true in positive curvature. When this condition is satisfied, we can use the convergent series expansion of $\ln(1-z_{\eta})=-\ln(1-1/\tilde z_{\eta})$ to write
\begin{multline}\label{boundone}|R_{N}| =
\Biggl|\tilde z_{\eta}^{N}\biggl(\sum_{s=1}^{N}\frac{\tilde z_{\eta}^{-s}}{s} + \ln (1-1/\tilde z_{\eta})\biggr)\Biggr| =
\Biggl|\sum_{s=N+1}^{\infty}\frac{\tilde z_{\eta}^{N-s}}{s}\Biggr|\\
\leq\frac{1}{N+1}\sum_{s=0}^{\infty}|\tilde z_{\eta}|^{-s-1} = \frac{1}{|\tilde z_{\eta}|-1}\frac{1}{N+1}\,\cvp
\end{multline}
which allows us to conclude.

\paragraph{Case two: $\tilde z_{\eta}=-1$.} Note that this can only occur in positive curvature. The reasoning used above when $|\tilde z_{\eta}|>1$ must be slightly refined, replacing the inequality in Eq.\ \eqref{boundone} by
\be\label{boundtwo} |R_{N}|=\Biggl|\sum_{s=1}^{N}\frac{(-1)^{s}}{s} + \ln 2\Biggr| =
\Biggl|\sum_{s=N+1}^{\infty}\frac{(-1)^{s}}{s}\Biggr| \underset{N\rightarrow\infty}{\sim}\frac{1}{2N}\,\cvp\ee
the large $N$ estimate being obtained by using the standard trick of grouping the terms of the alternating series two by two.

\paragraph{Case three: $|\tilde z_{\eta}|<1$.} Note that this can only occur in positive curvature. The term in $\tilde z_{\eta}^{N}\ln(1-z_{\eta})$ in $R_{N}$ is then exponentially small when $N\rightarrow\infty$. The result we want to prove will then follow from the estimate
\be\label{smallsigdef} \sigma_{N}(z) = z^{N}\sum_{s=1}^{N}\frac{z^{-s}}{s} \underset{N\rightarrow\infty}{\sim}\frac{1}{1-z}\frac{1}{N}\quad\text{when $|z|<1$.}\ee
A nice way to prove this is to use the following integral representation for $\sigma_{N}(z)$,
\be\label{sigNint} \sigma_{N}(z) =  z^{N}\int_{0}^{1/z}\frac{u^{N}-1}{u-1}\, \d u\, .\ee
For $|z|<1$, only the region very near the upper bound in the integration range can contribute at large $N$. Setting $u=\frac{1}{z}(1-v/N)$ and taking the large $N$ limit then yields
\be\label{sigNasy} \sigma_{N}(z)\underset{N\rightarrow\infty}{\sim} \frac{1}{z N}\int_{0}^{\infty}\frac{e^{-v}\d v}{1/z -1} = \frac{1}{1-z}\frac{1}{N}\,\cdotp\ee
\end{proof}

\begin{lemma}\label{lemmatwo} We have, when $N\rightarrow\infty$,
\be\label{SigNoneform22} \Sigma_{N}^{(2)} = -q(q+1)z_{\eta}\ln N  - \ln\prod_{k=1}^{q}\frac{\Gamma\bigl(1+q-\omega_{q,k}^{\eta}(r_{0})\bigr)}{k!}
+ O(1/N)\, .\ee
\end{lemma}
\begin{proof}
We are actually going to prove the identity
\be\label{SigNoneform2} \Sigma_{N}^{(2)} = \ln\prod_{k=1}^{q}\frac{k!\,\Gamma\bigl(N+1+q-\omega_{q,k}^{\eta}(r_{0})\bigr)}{(N+k)!\,\Gamma\bigl(1+q-\omega_{q,k}^{\eta}(r_{0})\bigr)}\, \cvp\ee
which is valid at finite $N$. Eq.\ \eqref{SigNoneform22} follows straightforwardly by using Stirling's formula
\be\label{Stirlingf} \ln\Gamma (N+1+a) = N\ln N - N +\Bigl(a+\frac{1}{2}\Bigr)\ln N + \frac{1}{2}\ln (2\pi) + O(1/N)\ee
and 
\be\label{sumroots} \sum_{k=1}^{q}\omega_{q,k}^{\eta} = q(q+1) z_{\eta} + \frac{1}{2}q(q-1)\, .\ee
This sum of the roots of the Meixner polynomial can be derived easily from Eq.\ \eqref{Meixnerid} and the series representation of the hypergeometric functions.

To prove Eq.\ \eqref{SigNoneform2}, we start by noting that
\begin{multline}\label{fneval}f_{n}^{\eta}(1;-M^{2}_{\eta,q}) = F(-q,q+1,n+1,z_{\eta}) \\= \frac{n!}{(n+q)!} z_{\eta}^{q}M_{q}\Bigl(n+q;-2q,\frac{z_{\eta}}{z_{\eta}-1}\Bigr)=
\frac{n! q!}{(n+q)!}P_{q}^{(n,-n)}(1-2z_{\eta})=\\
\frac{n!}{(n+q)!}\prod_{k=1}^{q}\bigl(n+q-\omega^{\eta}_{q,k}(r_{0})\bigr)\, ,
\end{multline}
where $P_{q}^{(n,-n)}$ is the Jacobi polynomial, generalizing \eqref{fnoteval}, and we have used the last equation in \eqref{Meixnerid} to derive the last equality. This last equality is actually all we need, but we have indicated the expression in terms of the Jacobi polynomials for completeness. Note that this is always strictly positive in negative curvature, as can be checked straightforwardly from the expansion of the hypergeometric function in the first line of the above equation and the fact that $z_{-}\leq 0$. We thus get, from the definition \eqref{Sigsum2},
\be\label{Sig2calc}\begin{split} \Sigma_{N}^{(2)}& =\sum_{n=1}^{N}\biggl[ 
\ln\frac{n!}{(n+q)!} + \ln\prod_{k=1}^{q}\bigl(n+q-\omega^{\eta}_{q,k}(r_{0})\bigr)\biggr]\\
& = \sum_{n=1}^{N}\ln\prod_{k=1}^{q}\frac{n+q-\omega^{\eta}_{q,k}(r_{0})}{n+k}= \ln\prod_{k=1}^{q}\prod_{n=1}^{N}\frac{n+q-\omega^{\eta}_{q,k}(r_{0})}{n+k}\, \cdotp
\end{split}\ee
Using $\Gamma(z+1) = z\Gamma(z)$, one can straightforwardly check that this is the same as Eq.\ \eqref{SigNoneform2}.
\end{proof}
Combining Eqs.\ \eqref{stra1}, \eqref{SigNoneform} and \eqref{SigNoneform22}, we get Eq.\ \eqref{toprove1}.
\end{proof}

\subsection{\label{specialSec2}Application: the large $\ell$ limit in negative curvature}

As a simple application of the explicit formula \eqref{specialdetform} for the determinant, let us study the large $\ell$ limit in negative curvature. As explained in Section \ref{infSec}, this limit is notoriously subtle. The logarithm $\ln D^{-}_{\infty}$ of the strictly infinite area determinant, per unit area, is given exactly, for any mass $M$, by Eq.\ \eqref{dethyperinf}. However, the infinite area limit of the finite area determinants does not coincide with the strictly infinite area result,
\be\label{infnotinf} \lim_{\ell\rightarrow\infty}\frac{1}{\ell}\ln D^{-}\not = D^{-}_{\infty}\, .\ee
This was explicitly checked when the mass is zero at the end of Section \ref{zeromassSec}. For masses of the form $M^{2} = q(q+1)$, $q\in\mathbb N$, the goal of the present subsection is to show the following.
\begin{proposition}
Let $a_{1},\ldots,a_{q}$ be the roots of the Laguerre polynomial $L_{q}^{(-2q -1)}$. Then the large $\ell$ asymptotics of the determinants computed for the special masses $M^{2}=q(q+1)$, $q\in\mathbb N$, is given by 
\be\label{largelfinal} \ln D^{-}\bigl(M^{2}=q(q+1),r_{0}(\ell)\bigr) = \biggl(-\frac{2}{3}+\sum_{k=1}^{q}a_{k}\ln (-a_{k})\biggr)\frac{\ell}{2\pi} + O(1)\, .\ee
\end{proposition}
This generalizes the leading term in the expansion at zero mass given in Eq.\ \eqref{largelzerom}. Note that the roots $a_{k}$ either have an imaginary part and then come in complex conjugate pairs (because the Laguerre polynomial $L_{q}^{(-2q -1)}$ has real coefficients) or are real and negative (because we know that $\ln D^{-}$ must be real, as stated just below Eq.\ \eqref{specialdetform}). The logarithm of the complex roots in \eqref{largelfinal} is evaluated with the determination $\ln (-a_{k}) = \ln |a_{k}| + i\arg (-a_{k})$, with $-\pi<\arg (-a_{k})<\pi$. For instance,
\be\label{exlargel} \begin{split}\ln D^{-}\bigl(M^{2}=2,r_{0}(\ell)\bigr) &= -\Bigl(\frac{1}{3}+\ln 2\Bigr)\frac{\ell}{\pi}\,\cvp\\
\ln D^{-}\bigl(M^{2}=6,r_{0}(\ell)\bigr) & =-\Bigl(\frac{2}{3}+3\ln 12-\frac{\pi}{\sqrt{3}}\Bigr)\frac{\ell}{2\pi}\, \cvp\quad\text{etc.}\end{split}\ee
\begin{proof}
Using \eqref{ellminrel} for $\ell^{-}=\ell\rightarrow +\infty$ and the fact that the Legendre polynomial $P_{q}$ is of degree $q$, it is straightforward, starting from \eqref{specialdetform}, to show that
\begin{multline}\label{llexp1} \ln D^{-}\bigl(M^{2}=q(q+1),r_{0}(\ell)\bigr) = \frac{q(q+1)}{2\pi}\ell\ln\frac{\ell}{4\pi} - \Bigl(\frac{1}{3\pi} + \frac{q(q+1)}{2\pi}\Bigr)\ell + q\ln\ell\\ - 2\ln\prod_{k=1}^{q}\Gamma\bigl(1+q-\omega^{-}_{q,k}(r_{0})\bigr) + O(1)\, .
\end{multline}
The non-trivial part of the calculation is to evaluate the asymptotic behaviour of the terms involving the $\Gamma$ functions. To do this, we need to find the asymptotic behaviour of the roots $\omega^{-}_{q,k}(r_{0})$ of the Meixner polynomials $M_{q}(x;-2q,r_{0}^{2})$. By examining the explicit expression of these polynomials given by the second equality in \eqref{Meixnerid}, and noting that, in the limit we are considering,
\be\label{zminusasy} z_{-}=-\frac{r_{0}^{2}}{1-r_{0}^{2}} = -\frac{\ell}{4\pi} + \frac{1}{2} + O\bigl(1/\ell\bigr)\, ,\ee
one finds 
\be\label{rootsexp} \omega_{k,q}^{-}=-a_k z_{-}+b_k+O\bigl(1/z_{-}\bigr) = \frac{a_{k}\ell}{4\pi} + b_{k}-\frac{1}{2}a_{k} + O\bigl(1/\ell\bigr)\ee
where the $a_{k}$ are the roots of the polynomial $$\sum_{s=0}^{q}\frac{(q+s)!}{s!(q-s)!}x^{q-s}\, ,$$ which is proportional to the Laguerre polynomial $L_{q}^{(-2q-1)}$. The $b_{k}$ (as well as the higher order terms in the large $\ell$ expansion) can be straightforwardly expressed in terms of the $a_{k}$, but we won't need the (rather complicated) explicit expressions. It is sufficient to know the sums
\be\label{sumsrasy} \sum_{k=1}^{q}a_{k} = -q(q+1)\, ,\quad \sum_{k=1}^{q}b_{k} = \frac{1}{2}q(q-1)\, ,\ee
which are obtained directly from Eq.\ \eqref{sumroots}. Using Eqs.\ \eqref{rootsexp} and \eqref{sumsrasy} together with the Stirling asymptotic expansion of the $\Gamma$ function, already indicated in Eq.\ \eqref{Stirlingf}, we find Eq.\eqref{largelfinal}.
\end{proof}
It is interesting to note that the terms in $\ell\ln\ell$ and $\ln\ell$ precisely cancel. The cancellation of the $\ell\ln\ell$ terms was actually expected, since it is crucial for the existence of an effective action per unit area, which is a consequence of extensivity.

\subsection*{Acknowledgments}

We would like to dedicate this paper to Steve Zelditch. F.F.\ shared with him illuminating discussions on many subjects, including in relation with a large part of the research presented in this paper. Steve was a magnificent mathematician with a deep interest in physics, a warm and caring human being, and a friend.

\subsection*{Declarations}

\subsubsection*{Funding}

This work is partially supported by the International Solvay Institutes and the Belgian Fonds National de la Recherche Scientifique FNRS (convention IISN 4.4503.15). The work of S.C.\ is supported by a postdoctoral research fellowship of the Belgian F.R.S.-FNRS.

\subsubsection*{Data Availability}

Data sharing not applicable to this article as no datasets were generated or analysed during
the current study.

\subsubsection*{Conflict of interest}

The authors have no relevant financial or non-financial interests to disclose.

\appendix\clearpage

\section{\label{AppA}Monotonicity of the eigenvalues as a function of the boundary length}

In this Appendix, we show that the eigenvalues $\la^{\pm}_{n,k}(r_{0})$ of the Laplacians $\Delta_{\pm}$, defined by Eqs.\ \eqref{Deltadef} and \eqref{conffactor}, are strictly decreasing functions of $r_{0}$. This result was mentioned in the main text in Section \ref{SLoneSec}. By taking the flat space limit, this also implies that the eigenvalues $\la^{0}_{n,k}(\ell)$ of $\Delta_{0}$ are strictly decreasing functions of the boundary length $\ell$.

The operator $\Delta_{\pm}$, with Dirichlet boundary conditions, are symmetric operators with respect to the scalar product
\be\label{scalproddisk} \langle \psi_{1} ,\psi_{2}\rangle = \int_{\disk}\d^{2}x\sqrt{g_{\pm}}\,  \psi_{1}\psi_{2}\, ,\ee
where $g_{\pm} = e^{2\sigma_{\pm}}(\d r^{2} + r^{2}\d\theta^{2})$. It is a positive operator because the Dirichlet boundary condition allows to perform an integration by part to prove that
\be\label{positiDelta} \langle \psi,\Delta_{\pm} \psi\rangle = \int_{\disk}\d^{2}x\sqrt{g_{\pm}}\, g_{\pm}^{\mu\nu}\partial_{\mu} \psi\partial_{\nu}\psi\geq 0\, .\ee
Thus, all the eigenvalues of $\Delta_{\pm}$ are positive. Eq.\ \eqref{positiDelta} actually shows that an eigenvector with zero eigenvalue must be a constant, which is impossible because the Dirichlet boundary condition would enforce this constant to be zero.

Let us denote by $(\psi^{\pm}_{k})_{k\geq 0}$ an orthonormal basis of eigenvectors, with eigenvalues $0<\la^{\pm}_{0}\leq\la^{\pm}_{1}\leq\cdots$ (this is just a rearrangement of the eigenvalues noted $\la_{n,k}^{\pm}$ above). The eigenvalues and the eigenfunctions are functions of the parameter $r_{0}$. It is convenient to rescale the radial coordinate and use $\rho=r_{0}r$ instead of $r$, so that the metrics
\be\label{metApp} g_{\pm} = \frac{4L^{2}}{(1\pm\rho^{2})^{2}}\bigl(\d\rho^{2} + \rho^{2}\d\theta^{2}\bigr)
\ee
do not depend explicitly on $r_{0}$ anymore. All the $r_{0}$ dependence is then in the range of the radial coordinate, $0\leq \rho\leq r_{0}$. We have
\be\label{eigenintform} \la^{\pm}_{k}(r_{0}) = \langle \psi^{\pm}_{k},\Delta_{\pm}\psi^{\pm}_{k}\rangle =\int_{0}
^{r_{0}}\d\rho\, \rho\int_{0}^{2\pi}\d\theta \, \Bigl[\bigl(\partial_{\rho}\psi^{\pm}_{k}\bigr)^{2} + \frac{1}{\rho^{2}}\bigl(\partial_{\theta}\psi^{\pm}_{k}\bigr)^{2}\Bigr]\, .\ee
The variation of $\la^{\pm}_{k}$ with respect to $r_{0}$ is thus
\begin{multline}\label{derlak1}
\frac{\d\la^{\pm}_{k}}{\d r_{0}}  = r_{0}\int_{0}^{2\pi}\d\theta\Bigl[\bigl(\partial_{\rho}\psi^{\pm}_{k}\bigr)^{2} + \frac{1}{\rho^{2}}\bigl(\partial_{\theta}\psi^{\pm}_{k}\bigr)^{2}\Bigr]\Big|_{r=r_0} \\+ \int_{0}^{r_{0}}\d \rho\, \rho\int_{0}^{2\pi}\d\theta \,\frac{\partial}{\partial r_{0}} \Bigl[\bigl(\partial_{\rho}\psi^{\pm}_{k}\bigr)^{2} + \frac{1}{\rho^{2}}\bigl(\partial_{\theta}\psi^{\pm}_{k}\bigr)^{2}\Bigr]\, .
\end{multline}
The second term can be massaged as follows:
\be\label{secndtermtrans}
\begin{split}
\int_{0}^{r_{0}}\d \rho\, \rho\int_{0}^{2\pi}\d\theta \,\frac{\partial}{\partial r_{0}} &\Bigl[\bigl(\partial_{\rho}\psi^{\pm}_{k}\bigr)^{2} + \frac{1}{\rho^{2}}\bigl(\partial_{\theta}\psi^{\pm}_{k}\bigr)^{2}\Bigr] = \int_{\disk}\d^{2}x\sqrt{g_{\pm}}\frac{\partial}{\partial r_{0}}\bigl(g_{\pm}^{\mu\nu}\partial_{\mu}\psi^{\pm}_{k}\partial_{\nu}\psi^{\pm}_{k}\bigr)\\
&=2\int_{\disk}\d^{2}x\sqrt{g_{\pm}}\, g_{\pm}^{\mu\nu}\partial_{\mu}\psi^{\pm}_{k}\partial_{\nu}\frac{\partial \psi^{\pm}_{k}}{\partial r_{0}}
 =2 \int_{\disk}\d^{2}x\sqrt{g_{\pm}}\, \psi^{\pm}_{k}\Delta_{\pm}\frac{\partial \psi^{\pm}_{k}}{\partial r_{0}}\\
& =2\int_{\disk}\d^{2}x\sqrt{g_{\pm}}\, \psi^{\pm}_{k}\frac{\partial}{\partial r_{0}}\Delta_{\pm}\psi^{\pm}_{k}
 =2 \int_{\disk}\d^{2}x\sqrt{g_{\pm}}\, \psi^{\pm}_{k}\frac{\partial}{\partial r_{0}}\bigl(\lambda^{\pm}_k\psi^{\pm}_{k}\bigr)\\
&= 2\frac{\d\la^{\pm}_{k}}{\d r_{0}} +2 \la^{\pm}_{k} \int_{\disk}\d^{2}x\sqrt{g_{\pm}}\, \psi^{\pm}_{k}\frac{\partial \psi^{\pm}_{k}}{\partial r_{0}} \, \cdotp
\end{split}\ee
Note that, to derive the third equality, we have used an integration by part and the Dirichlet boundary condition on $\psi^{\pm}_{k}$; and to derive the last equality we have used the normalisation condition
\be\label{normfk} \langle \psi^{\pm}_{k},\psi^{\pm}_{k}\rangle = \int_{\disk}\d^{2}x\sqrt{g_{\pm}}\, (\psi^{\pm}_{k})^{2} = 1\, .\ee
If we take the derivative of this condition with respect to $r_{0}$, we get a boundary term, that vanishes because $\psi^{\pm}_{k}$ is zero on the boundary, plus precisely the second term in the last line of \eqref{secndtermtrans}. Thus, this term vanishes and we obtain
\be \label{secndtermtransfinal} 
\int_{0}^{r_{0}}\d \rho\, \rho\int_{0}^{2\pi}\d\theta \,\frac{\partial}{\partial r_{0}} \Bigl[\bigl(\partial_{\rho}\psi^{\pm}_{k}\bigr)^{2} + \frac{1}{\rho^{2}}\bigl(\partial_{\theta}\psi^{\pm}_{k}\bigr)^{2}\Bigr] = 2\frac{\d\la^{\pm}_{k}}{\d r_{0}}\, \cdotp\ee
Using this result in Eq.\ \eqref{derlak1}, we get
\be\label{derlak2}
\frac{\d\la^{\pm}_{k}}{\d r_{0}}  = - r_{0}\int_{0}^{2\pi}\!\d\theta\,\bigl(\partial_{\rho}\psi^{\pm}_{k}\bigr)^{2}_{|\rho=r_{0}} <0\, .\ee
We have used the fact that $\partial_{\theta}\psi_{k}^{\pm}=0$ and $(\partial_{\rho}\psi_{k}^{\pm})^{2}>0$ along the boundary, which are direct consequences of the explicit form of the eigenfunctions, see Eqs.\ \eqref{eigenminus} and \eqref{eigenplus}.

\end{document}